\pdfoutput=1
\documentclass[12pt]{article}
\usepackage{amsmath, amsthm, amssymb}

\usepackage{graphicx}

\usepackage{color}
\definecolor{light}{gray}{.50}
 
\begin{document}


\newtheorem{thm}{Theorem}
\newtheorem{cor}{Corollary}
\newtheorem{lem}{Lemma}
\newtheorem{prop}{Proposition}

\newcommand{\noprint}[1]{}

\def\I{\mathbf 1}
\def\K{\mathbf K}
\def\R{\mathbf R}
\def\U{\mathbf U}
\def\X{\mathbf X}
\def\Y{\mathbf Y}

\def\0{\mathbf 0}
\def\1{\mathbf 1}

\def\b{\mathbf b}
\def\h{\mathbf h}
\def\r{\mathbf r}
\def\s{\mathbf s}
\def\t{\mathbf t}
\def\u{\mathbf u}
\def\v{\mathbf v}
\def\w{\mathbf w}
\def\x{\mathbf x}
\def\y{\mathbf y}

\def\one{\mathbf 1}
\def\verr{\mathbf\varepsilon}

\def\F{\mathcal F}
\def\H{\mathcal H}
\def\N{\mathcal N}
\def\P{\mathcal P}
\def\S{\mathcal S}

\def\rt{\rightarrow}
\title{Local additive estimation}
\date\today

\author{ Juhyun Park and Burkhardt Seifert \thanks{Juhyun Park was Postdoctoral Research
Fellow (Email:{\it juhyun.park@lancaster.ac.uk}), Burkhardt Seifert is Professor
(Email:{\it seifert@ifspm.uzh.ch}) in Biostatistics unit, Institute for Social and
Preventive Medicine, University of Z\"urich, Z\"urich, Switzerland. Funding for this work
was provided by Swiss National Science Foundation grant 20020-103743.} \\ Lancaster
University and University of Z\"urich}
\date{\today}

\maketitle

\begin{abstract}
Additive models are popular in high--dimensional regression problems because of
flexibility in model building and optimality in additive function estimation.
Moreover, they do not suffer from the so-called {\it curse
of dimensionality} generally arising in nonparametric regression setting. Less known is
the model bias incurring from the restriction to the additive class of models. 
We introduce a new class of estimators that reduces additive model bias and at the same
time preserves some stability of the additive estimator. This estimator is shown to
partially relieve the dimensionality problem as well. The new estimator is constructed by
localizing the assumption of additivity and thus named {\it local additive estimator}.
Implementation can be easily made with any standard software for additive regression. For
detailed analysis we explicitly use the smooth backfitting estimator by Mammen, Linton
and Nielsen~(1999). 

\end{abstract}

\section{Introduction} \label{sec:intro}


Application of additive models is numerous from econometrics, social sciences to
environmental sciences (Deaton and Muellbauer 1980; Hastie and Tibshirani
1990). Separability of each component is well suited for flexible and
interpretable model building in modern high dimensional problems with many
covariates.
The main advantage of additive regression is that it allows us to deal with
high-dimensional regression in one-dimensional precision.

Since the recognition of potential of additive models in 80s, several additive
estimators have been developed in various contexts of smoothing.
Earlier methods tend to be more algorithmic in nature because of nontrivial analyses
required to understand the behaviour of estimators (see Opsomer and Ruppert 1997; Opsomer
2000). 
More recent methods include marginal integration by Linton and
Nielsen~(1995) and smooth backfitting by Mammen, Linton and Nielsen~(1999). The smooth
backfitting estimator
(SBE) is shown to be oracle optimal for the additive function estimation, that is, it
achieves the same precision as in one--dimensional regression. The SBE is also applicable
when additivity is only approximately valid by means of a projection idea (Mammen
et al.~2001). 

Less known is the model bias incurring from the restriction to the additive class of
models. Additive models miss important (nonadditive) features by considering the
nonadditive part nuisance or noise.
This is also related to the fact that fitting additive models and diagnostics are less
trivial in that it involves various issues concerning model selection and
stability~(Breiman 1993).

Models without additive restriction fall in the broad category of nonparametric regression
models. Their properties have been well established in several earlier works, one of which
points out that local linear estimator is minimax optimal in more than one--dimensional
regression problem (Fan et al.~1997). However, as the dimension of the variables grows,
the stability of the estimation becomes increasingly an issue, which brings about {\it
curse of dimensionality} (see, e.g. Stone 1980, 1982). 

This situation leads to the question whether or how to combine advantages of those
estimators, the stability of additive estimator and the optimality of local linear one.
The approach proposed in Studer et al.~(2005) uses penalty to the nonadditive part,
which produces a family of {\it regularised} estimators. In this paper, we introduce
another class of estimators by {\it localizing} the additivity assumption and this
will be named {\it local additive estimator}.

Let $(\X,Y)$ be random variables of dimensions $d$ and 1, respectively and let
$(\X_i,Y_i), i=1, \cdots, n,$ be independent and identically distributed random variables
from $(\X,Y)$. Denote the design density of $\X$ by $f(\x)$. We assume that $\X$ has
compact support $[-1,1]^d$. The regression function $r(\x) = E[Y|\X = \x]$ is assumed to
be smooth.
The additive model has the relation
\begin{equation} \label{eq:radd}
r(\x) = r_{0} + r_{1}(x_1) + \cdots + r_{d}(x_d) \,.
\end{equation}
This is a global assumption on the shape of the regression function and thus quite
restrictive. 

Given $\x$, consider a $\w$-neighborhood of $\x$. If $||\w||$ is small enough, by Taylor theorem, we would have
\[
r(\x) \approx r_0 + r_1(x_1) + \cdots + r_d(x_d) \,.
\]
Note that this is not an {\it assumption} on the model. The accuracy of the approximation clearly depends on the $\w$-neighborhood. We will call this approximate additive relation {\it local additivity}.  

The above argument naturally leads to an estimator that can be constructed from additive estimator using data in the neighborhood of interest. 
For a given point $\x_0$, construct an additive estimator using data in the $\w$-neighborhood of $\x_0$. 
The new estimator is defined as the predictor of the additive estimator at $\x=\x_0$. 
This will be termed {\it local additive estimator}, denoted by $\hat r_{ladd}(\x_0)$. A
formal definition is given in Section~\ref{sec:main}.

By not directly imposing the additive restriction, we reduce model bias. On the other
hand, the merit of additivity that allows us to deal with high--dimensional regression in
one--dimensional precision is partially lost. The main advantages of the new estimator can
be summarized as follows. 1) Additivity is approximately valid {\it locally} even when the
true regression function is not additive. This helps keep bias small for general
regression function.
2) The local additive approximation is more flexible than the local linear one. 
Thus, the local region for the additive estimator can be chosen larger than that for the local linear one, which improves variance of the estimator. 
3) Standard software for additive estimators is directly applicable.

The paper is organized as follows. We formulate main results in
Section~\ref{sec:main}, followed by asymptotic
comparison to the local linear estimator, $\hat r_{ll}$, and the additive estimator, $\hat
r_{add}$ as an illustration. Smoothing parameter selection is also discussed. Numerical
studies are found in Section~\ref{sec:numeric} with an application to a real data example.
An extended version of simulation studies and some proofs of Section~\ref{subsec:sbe} are
found in Park and Seifert~(2008).










\section{Local additive estimation} \label{sec:main}

\subsection{Preliminaries}

Let $\x_0$ be a fixed interior output point. For $\w=(w_1,\cdots,w_d)$, we apply an additive estimator $\hat r_{add}$ using data in a $\w$-neighborhood of $\x_0$. 
Our analysis is based on $d$--dimensional rectangular region $[\x_0 \pm \w] = \{\X_i, \X_i
\in [\x_0 -\w, \x_0+\w]\}$. 
Denote the number of observations $\X_i$ in $[\x_0 \pm \w]$ by $\tilde n$. 
Properties of the local additive estimator can be developed by rescaling the region $[\x_0
\pm \w]$ to $[-1,1]^d$ and then using results known for $\hat r_{add}$. 
We will consider additive estimators that reach the optimal order $O(n^{-4/5})$.
For technical reasons, we will focus on linear estimators, which enable us to compute expectations under Taylor expansions.
The SBE by Mammen et al.~(1999) is known to be oracle optimal under general conditions,
and other estimators inherit this optimality under more special situations (Linton and
Nielsen 1995; Opsomer and Ruppert 1997; Opsomer 2000).
Throughout the article, we will assume that

\begin{itemize}
\item[$(A.1)$] The regression function $r$ and the design density $f$ are twice
continuously
differentiable.
\end{itemize}
The special case of uniform design will be separately dealt with later in this section.

When additive estimator is viewed as a componentwise one--dimensional smoother, it has
inherently a smoothing parameter associated with it. 
It may refer to smoothing window $h$ as in kernel smoothers, smoothing parameter $\lambda$ as in smoothing splines, or generally degrees of freedom $df$ as in equivalent linear smoothers. 
We will stick to $h$ for a smoothing parameter, as the local linear smoother is used later in our analysis. 
 
Suppose that all $w_j$'s are of same order. For simplicity of notation let $w_j=w$. Let $w
\rt 0$ and $h_j/w \rt 0$. Write 
\begin{equation}\label{e:u}
\U = \frac{\X-\x_0}{w} \,,
\end{equation}
for the rescaled random variable on $[-1,1]^d$ with density
\begin{eqnarray}
\tilde f(\u)  
&=& f(\x_0+w\u)/\int_{[-1,1]^d} f(\x_0+w\u)\,d\u \nonumber \\
&=& \frac{f(\x_0+w\u)}{2^df(\x_0)} + O(w^2) \label{e:ftilde}\,.
\end{eqnarray}
The corresponding regression function is
\begin{equation} \label{e:tr}
 \tilde r(\u) = r(\x_0+w\u)
\end{equation}
and the transformed bandwidth is
\begin{equation}
\tilde h_j = h_j/w \label{e:th} \,.
\end{equation}
The local additive estimator at $\x_0$ is defined as
$\hat r_{ladd}(\x_0)=\widehat{\tilde r}_{add}(\0)$.

Denote 1st and 2nd partial derivatives of $r$ by $r_j'(\x)$, $r_{j,k}''(\x)$ and
the $d\times d$ matrix of 2nd derivatives by $\r''$. 
$\hat r_{ll}(\x_0)$ and the local additive estimator by $\hat r_{ladd}(\x_0)$.
We write E, B, V, MSE, ISE, ASE, MISE and MASE for the conditional expectation, bias,
variance, mean squared error, integrated squared error, average squared error, integrated
mean squared error and average mean squared error, respectively. 
Define a matrix norm $||\cdot||$ for a symmetric matrix $A = \{a_{ij}\}$ as $||A|| = max_{i,j}|a_{ij}|$ and write $||\cdot||_2$ for the usual $L_2$ norm. 

Let us first consider a bilinear function of components $u_j$ and $u_k$ as
\[
b^{jk}(\u) = (u_j-\bar U_j)(u_k-\bar U_k)\,,
\]
where $\bar U_j$ and $\bar U_k$ are $j$th and $k$th marginal averages of $\U$ in
(\ref{e:u}). Note that
$\bar U_j$ and $\bar U_k$ are considered constants given $\U$. 
We will see that studying this function is revealing when applying Taylor expansions in
the proof of our main results.
Let $f_w$ be a sequence of design densities that converges to uniform. This can be
constructed, for example, as in (\ref{e:ftilde}) by defining $f_w(\u) =
\frac{f(\x_0+w\u)}{2^df(\x_0)}$ for a density $f$ satisfying (A.1). Let $\hat
b^{jk}_{add,w}$ be the corresponding additive estimator.
If $u_j$ and $u_k$ are uniformly  distributed, as $n \rt\infty$, $b^{jk}(\0) \rt 0$ and
$\hat b^{jk}_{add,0}\rt 0$. Thus,
$\hat b^{jk}_{add,w}(\0)$ should converge to zero too. 
Surprisingly enough, the case of vanishing second partial derivatives needs special attention.
Denote
\begin{equation}\label{A}
A_{j,k} =\left\{ \x \in [-1,1]^d \mid r_{j,k}''(\x) =0 \right\} \,.
\end{equation}
Without higher order smoothness assumption, the results below are only valid for $\x_0$ outside the borders $\partial A_{j,k}$ of $A_{j,k}$.
We claim however that these borders are small and can be ignored for most practical
situations, as explained in the remarks following Proposition~\ref{prop:ladd_unif_b2} in
Section~\ref{subsec:sbe}.

In addition to $(A.1)$, the following assumptions are made. 

\begin{itemize}
\item[$(A.2)$] The kernel $K$ is bounded, has compact support, is symmetric around 0 and is Lipschitz continuous.

\item[$(A.3)$] The density $f$ of $\x$ is bounded away from zero and infinity on $[-1,1]^d$. 

\item[$(A.4)$] For some $\theta > 5/2, E[|Y|^\theta] < \infty$.

\item[$(A.5)$] $\tilde h_j \rt 0$ such that
$\tilde n\tilde h_j^d/\ln \tilde n \rt \infty$ as $\tilde n\rt\infty$.
\end{itemize}

\subsection{Main result}

\begin{thm} \label{thm:main}
Assume that $\hat r_{add}$ is linear in $Y$ and oracle optimal.
Let $f_w$ be a sequence of design densities that converges to uniform $f_0$ and
$\hat r_{add,w}$ be the corresponding additive estimator.
Assume that $\hat r_{add,w}$ converges as $f_w$ converges and satisfies
\[
|\hat b^{jk}_{add,w}(\0) - \hat b^{jk}_{add,0}(\0)| \leq L||f_w-f_0||_2^2\,\mbox{ for all } j\neq k\,,
\]
where $L$ is a constant.
Then, for all $\x_0 \not\in \bigcup_{j,k} \partial A_{j,k}$ defined in (\ref{A}),
\begin{eqnarray*}
B^2[\hat r_{ladd}(\x_0)] &=& max\{O(h^4), O(w^8 + w^4\max_{j,k}|\hat b^{jk}_{add,0}(\0)|^2)\}\\
V[\hat r_{ladd}(\x_0)] &=& O((nw^{d-1}h)^{-1}) \,.
\end{eqnarray*}
\end{thm}

\begin{proof}
Here, we will present the main ideas for bias.
Because the estimator is linear, we have
\[
E[\hat r_{add}(\x_0)] = \frac{1}{n}\sum_{i=1}^nW_i(\x_0,\X_i)r(\X_i) \,.
\]
Similarly, for the local additive estimator, we have
\[
E[\hat r_{add,w}(\x_0)] = \frac{1}{\tilde n}\sum_{i=1}^{\tilde n} \tilde W_i(\0,
\U_i)\tilde r(\U_i)\,,
\]
where $\U$ is given in (\ref{e:u}).
\begin{eqnarray*}
\tilde r(\U_i) &=& r(\x_0 +w\U_i) = r(\x_0) +w\sum_j r_j'(\x_0)U_{ij} + \frac{w^2}{2}\sum_{j,k} r_{j,k}''(\x_0)U_{ij}U_{ik} + R(\x_0, \U_i)\\
&=& additive + \frac{w^2}{2}\sum_{j\neq k}r_{j,k}''(\x_0)U_{ij}U_{ik} + R(\x_0, \U_i) \,.
\end{eqnarray*}
Thus,
\begin{eqnarray*}
B[\hat r_{ladd}(\x_0)]&=& \frac{1}{\tilde n}\sum_{i=1}^{\tilde n}\tilde W_i(\0,\U_i)\tilde
r(\U_i) - r(\x_0) \\
&=& B[additive] + \frac{w^2}{2}\sum_{j\neq k} r_{j,k}''(\x_0)\Big(\frac{1}{\tilde
n}\sum_{i=1}^{\tilde n}\tilde W_i(\0,\U_i)U_{ij}U_{ik}\Big) \\
& & + \frac{1}{\tilde n}\sum_{i=1}^{\tilde n}\tilde W_i(\0,\U_i)R(\x_0,\U_i)\,.
\end{eqnarray*}
Because of oracle optimality of the estimator, the bias of the additive part becomes 
\begin{eqnarray}
B[additive] &=& \frac{\tilde h^2}{2}\Big(\frac{w^2}{2}\sum_{j}2r_{j,j}''(\x_0)\Big) + o(\tilde h^2w^2) \nonumber \\
&=& \frac{h^2}{2}\sum_j r_{j,j}''(\x_0) + o(h^2) \label{e:a1}\,,
\end{eqnarray}
the latter equality following from (\ref{e:th}).
For the leading nonadditive term, first consider
\[
\frac{1}{\tilde n}\sum_{i=1}^{\tilde n}\tilde W_i(\0,\U_i)U_{ij}U_{ik} \,.
\]
Observe that
\begin{equation} \label{e:bjk}
U_{ij}U_{ik} = (U_{ij}-\bar U_j)(U_{ik}-\bar U_k) + \bar U_jU_{ik}+\bar U_kU_{ij}+\bar U_j\bar U_k \,.
\end{equation}
Given $\U_i$, the last three terms are linear and thus do not add additional bias. Therefore, we focus on 
\[
\frac{1}{\tilde n}\sum_{i=1}^{\tilde n}\tilde W_i(\0,\U_i)(U_{ij}-\bar U_j)(U_{ik}-\bar
U_k) = \hat b^{jk}_{add,w}(\0) \,.
\]
This is nothing but the additive estimator at $\0$ when the design density is $f_w$ and the true regression function is the bilinear function $b^{jk}$.
It may be written as
\[
\hat b^{jk}_{add,w}(\0) = \hat b^{jk}_{add,w}(\0) -\hat b^{jk}_{add,0}(\0) + \hat b^{jk}_{add,0}(\0) \,.
\]
Thus,
\begin{eqnarray}
|\hat b^{jk}_{add,w}(\0)| &\leq& L||f_w-f_0||_2^2 + |\hat b^{jk}_{add,0}(\0)| \nonumber\\
&=& O(w^2 + |\hat b^{jk}_{add,0}(\0)|) \,. \label{e:a2}
\end{eqnarray}
Therefore, the second term is of order $O(w^2)O(w^2 + |\hat b^{jk}_{add,0}(\0)|)$.
The last remainder term may be written as
\[
w^2\sum_{j,k}\Big(\frac{1}{\tilde n}\sum_{i=1}^{\tilde n}\tilde
W_i(\0,\U_i)U_{ij}U_{ik}\int_0^1(1-\theta)\{r_{j,k}''(\x_0+\theta
w\U_i)-r_{j,k}''(\x_0)\}\,d\theta \Big)\,.
\]
As $\r''$ is continuous, the integrands are $o(1)$ and the corresponding terms become negligible compared to the main term above, if $r_{j,k}''(\x_0) \not= 0$.
If $r_{j,k}''(\x)=0$ in a neighborhood of $\x_0$, the corresponding integrand vanishes.
Hence, the result follows from (\ref{e:a1}) and (\ref{e:a2}).
\end{proof}

To demonstrate the idea of our result, we make a rough comparison to the existing results
in the following two sections by differentiating a situation
with additive regression
function from that with general regression function.

\subsection{Behavior for additive regression function}

When the true regression function is additive, the additive estimator $\hat r_{add}$ has
MSE
of $O(n^{-4/5})$ and the local linear estimator $\hat r_{ll}$ has MSE of
$O(n^{-4/(4+d)})$. 
We can see this from
\begin{eqnarray*}
V\big[\hat r_{ll}(\x_0)\big] = O\big((nh^d)^{-1}\big) \,, & &B^2\big[\hat
r_{ll}(\x_0)\big] = O\big(h^4||\r''||^2\big) = O\big(h^4\big) \,,\\
V\big[\hat r_{add}(\x_0)\big] = O\big((nh)^{-1}\big) \,,& & B^2\big[\hat
r_{add}(\x_0)\big] = O\big(h^4(||\r''||)^2\big) = O(h^4) \,.
\end{eqnarray*}
The local additive estimator $\hat r_{ladd}$ should beat the local linear estimator
and come as close to the additive one as possible.
With the same principle, the local additive estimator would have
\begin{eqnarray*}
V\big[\hat r_{ladd}(\x_0)\big] &=& O\big((\tilde n\tilde h)^{-1}\big) =
O\big((nw^{d-1}h)^{-1}\big) \,,\\
B^2\big[\hat r_{ladd}(\x_0)\big] &=& O\big(\tilde h^4(||\tilde\r''||)^2\big) = O(h^4) \,.
\end{eqnarray*}
Obviously, the additive estimator is optimal, the local linear estimator is worst, and the
local additive estimator is in between.

\subsection{Behavior for general regression function}

Now consider the general case. Note that properties of additive estimators for general
regression functions are not well studied. 
Nevertheless, when the true regression function is not additive, bias of the additive
estimator is $O(1)$. 
Variance does not depend on the regression function and thus remains the same.
Thus we have
\begin{eqnarray*}
V\big[\hat r_{ll}(\x_0)\big] = O\big((nh^d)^{-1}\big) \,, & &B^2\big[\hat
r_{ll}(\x_0)\big] = O\big(h^4||\r''||)^2\big) = O\big(h^4\big) \,,\\
V\big[\hat r_{add}(\x_0)\big] = O\big((nh)^{-1}\big) \,,& & B^2\big[\hat
r_{add}(\x_0)\big] = O\big(||\r''||^2\big) = O(1) \,.
\end{eqnarray*}

Applying the same principle to the local additive estimator would lead to
\begin{eqnarray*}
V\big[\hat r_{ladd}(\x_0)\big] &=& O\big((\tilde n\tilde h)^{-1}\big) =
O\big((nw^{d-1}h)^{-1}\big) \,,\\
B^2\big[\hat r_{ladd}(\x_0)\big] &=& O\big(||\tilde \r''||^2\big) = O(w^4) \,.
\end{eqnarray*}
We will show (Theorem 2) that the limit for the bias of $\hat r_{ladd}(\x_0)$ can be
further improved to $B^2[\hat r_{ladd}(\x_0)] = O(w^8)$ using the SBE. 

\subsection{Local additive estimator based on the SBE} \label{subsec:sbe}

When the regression function is additive, it can be shown that there is no
loss in bias with local additive estimator compared to additive estimator. For
general case, the local additive estimator based on the SBE satisfies the requirements of
Theorem 1. Note that for the SBE,
existence and convergence occur with probability tending to one (see Mammen et al.~1999),
thus our statements imply the same without explicitly mentioning it.


The results are valid under quite general distributions, see assumption (A.4). 
For simplicity of notation we will assume that the residuals $\varepsilon$ have
constant variance $\sigma^2$ whenever appropriate.

\begin{thm} \label{thm:main2}
The local additive estimator $\hat r_{ladd}$ based on the smooth backfitting estimator fulfills Theorem~\ref{thm:main} with $\hat b^{jk}_{add,w}=O(w^2)$ and
\[
V[\hat r_{ladd}(\x_0)] = 2\mu_0(K^2)\sigma^2\sum_{j=1}^d\left(nw^{d-1}h_j\right)^{-1}(1+o(1))\,.
\] 
\end{thm}


\begin{cor} \label{cor:cor1}
For all $\x_0 \not\in \bigcup_{j,k} \partial A_{j,k}$ defined in (\ref{A}), the local additive estimator $\hat r_{ladd}$ based on the smooth backfitting estimator has $B^2[\hat r_{ladd}(\x_0)] = \max\{O(h^4), O(w^8)\}$. 
\end{cor} 
In brief, the projection property of the SBE together with (A.1) helps reducing the bias
for the general regression function.
In summary we have for general regression function
\begin{equation} \label{e:mse}
MSE[\hat r_{ladd}(\x_0)] = O(h^4 + w^8 + (nw^{d-1}h)^{-1})\,.
\end{equation}
 
\begin{cor} \label{cor:cor2}
Assume that $d\leq 8$. Optimal orders of $w$ and $h$ of the local
additive estimator
$\hat r_{ladd}$ based on the smooth backfitting estimator are given by
\[
w \sim n^{-1/(9+d)} \,,\quad h \sim n^{-2/(9+d)} \,,
\]
leading to
\[
MSE[\hat r_{ladd}(\x_0)] \sim n^{-8/(9+d)} = n^{-4/\big(4+\frac{d+1}{2}\big)} \,.
\]
\end{cor}

\noprint{
\begin{proof} To select MSE--optimal smoothing parameters, observe that the minimum occurs when $h^4$ and $w^8$ have the same order.
Hence, the optimal relation between $h$ and $w$ is $h \sim w^2$ and thus we have
\[
MSE[\hat r_{ladd}(\x_0)] = O\big(w^8 + (nw^{d+1})^{-1}\big)\,.
\]
Minimization of MSE with respect to $w$ leads to the result.
\end{proof}
}

In comparison, the optimal local linear estimator achieves $O(n^{-4/(4+d)})$. 
The reduction of dimensionality is explained by the factor $\tilde
d=\frac{d+1}{2}$, the {\it equivalent dimension}. For example when $d=3$ the local
additive estimator behaves similar to a local linear estimator with $\tilde d=2$, and
when $d=5$ it will be reduced to $\tilde d=3$. Thus, local additive estimation provides
some relaxation of
dimensionality in nonparametric regression compared to the minimax local linear
estimator.

It turns out that the existence of second derivatives is not sufficient to derive explicit
coefficients for leading terms.
Below we deal with the special situation of a uniform design with higher order
smoothness assumption.

\begin{itemize}
\item[$(A.1^\prime)$] The regression function $r$ is four times continuously differentiable and $f$ is uniform.
\end{itemize}

\begin{prop} \label{prop:ladd_unif_b2}
Suppose that ($A.1^\prime$) holds. Bias of the local additive estimator $\hat r_{ladd}$ based on the smooth backfitting estimator is given by
\[
B[\hat r_{ladd}(\x_0)]=\Big(\frac{\mu_2(K)}{2}\sum_{j=1}^d h_j^2r_{j,j}''(\x_0)-\frac{w^4}{4!\cdot 9}\sum_{j\neq k} r_{j,j,k,k}''''(\x_0)\Big)+o\left(h^2+w^4\right)\,.
\]
\end{prop}

Contrary to Theorems \ref{thm:main} and \ref{thm:main2}, the Proposition is valid
without any exclusion of boundaries $\partial A_{j,k}$, which implies that the restriction
there is
related to irregular points of the regression function only. It should be mentioned
however that irrespective of condition $(A.1^\prime)$ the MSE is always of
order
$O\left(h^4+w^6+(nw^{d-1}h)^{-1}\right)$ if $r''$ is Lipschitz continuous. Thus,
the local additive estimator {\it works} also at the remaining boundaries.
Proposition~\ref{prop:ladd_unif_b2} additionally shows that higher order smoothness
assumption would not help further reduce bias.
Moreover, it can be deduced from the proof (not shown) that the existence of $\r''$
is not sufficient to derive leading terms. 

The optimal smoothing parameters are determined in the following.
Define
\[
a = \frac{\mu_2(K)}{2}\sum_{j}r_{j,j}''(\x_0) \,,\quad
b = \frac{1}{4!\cdot 9}\sum_{j\neq k} r_{j,j,k,k}''''(\x_0)\,,\quad c = 2d\mu_0(K^2)\sigma^2 \,.
\]
\begin{prop} \label{prop:ladd_unif_w}
Suppose that ($A.1^\prime$) holds. Assume that $h_j=h$ and let $h=C_hw^2$. The
smoothing parameter $w$ that minimizes asymptotic MSE is given by
\[
w = \left(\frac{c(d+1)}{8C_h(aC_h^2-b)^2}\right)^{1/(9+d)}n^{-1/(9+d)} \,.
\]
\end{prop}

\begin{prop} \label{prop:ladd_unif_wconst}
Under the same assumptions as in Proposition~\ref{prop:ladd_unif_w}, the optimal choice of $C_h$ is given by
\[
C_h = \sqrt{\frac{2}{d-1}\Big(-\frac{b}{a}\Big)} \,.
\]
provided that $ab<0$.
\end{prop}

Properties of the local additive estimator based on
the SBE are studied in detail in Park and Seifert~(2008). Proofs of
Propositions~\ref{prop:ladd_unif_b2}--\ref{prop:ladd_unif_wconst} are found there
and results of Theorem~\ref{thm:main2} and Corollary~\ref{cor:cor1} can be deduced
directly from results formulated there.

\subsection{Data-adaptive parameter selection}\label{subsubsec:aic}

We consider smoothing parameter selection based on model selection criteria for general
regression function estimation.

Although asymptotic equivalence of classical model selection criteria has long been
recognized (H\"ardle et al.~1988), because of small sample behavior, several versions
of model selection criteria exist (Hurvich and Simonoff 1998). Still most discussions were
limited to one dimensional problem.

For additive models with ordinary backfitting estimator, Opsomer and Ruppert (1998)
proposed a plug-in bandwidth selector and Wood~(2000) proposed generalized
cross-validation approach for additive models with penalized regression splines. For
additive models with smooth backfitting estimator, Nielsen and Sperlich~(2005) discussed
cross validation, while Mammen and Park~(2005) proposed a bandwidth selection method
which minimizes a penalized sum of squared residuals 
\[
PLS = \hat\sigma^2\left(1+2\,\sum_{j}\frac{1}{nh_j} K(0) \right)
\]
and noted that it is computationally more feasible than cross validation. 
They also conjectured about model misspecification (p.~1263) that {\it ...the penalized
least squares bandwidth will work reliably also under misspecification of the additive
model. This conjecture is supported by the definition of
this bandwidth...} but pointed out the difficulty involved in the theory (p.~1267).

For nonadditive models Studer et al.~(2005), in the context of penalized additive
regression approach, investigated parameter selection based on AIC-type model selection
criteria such as AIC, GCV, and AIC$_C$~(Hurvich et al.~1998) and established asymptotic
equivalence of these estimators in multivariate local linear regression for $d\leq 4$
where the estimator satisfies stability condition. 
Note that the additive SBE uses only two-dimensional marginal densities and thus such
restriction is not necessary.

We investigate smoothing parameter selection based on AIC-type model selection criteria
and show that PLS is equivalent to AIC-type model selection criteria.
Because the local additive estimator based on the SBE uses two-dimensional densities in
the rescaled window, the formulas (6.18)-(6.21) in Mammen and Park~(2005) can be used to
show that (A.5) is sufficient for the local additive estimator to be stable. In view of
Corollary~\ref{cor:cor2}, (A.5) is necessary too.

Consider 
\[
AIC(h,w) = \log(\hat\sigma^2) + 2tr(H)/n \,,
\]
where $\hat\sigma^2 = \frac{1}{n}||\Y-H\Y||^2$, $\Y$ is the column vector
of
responses on design points with a hat matrix $H$ and $tr(H)$ is the trace of the hat
matrix $H$. 
Using
\[
\log(\hat\sigma^2) = \log(\sigma^2) + \frac{\hat\sigma^2}{\sigma^2} - 1 +
O_p\left((\hat\sigma^2-\sigma^2)^2\right) \,.
\]
Studer et al.~(2005) defined the Taylor approximation of AIC$-$ $\log(\sigma^2)$
by
\begin{equation}\label{e:aic}
\mbox{AIC}_T = \frac{\hat\sigma^2}{\sigma^2} - 1 + \frac{2}{n}tr(H) \,.
\end{equation}
It can be shown that AIC and AIC$_T$ are equivalent for the optimal parameters in
Corollary~\ref{cor:cor2}. Using the fact that for additive regression functions
\[
tr(H) \rt K(0)\sum_{j}1/h_j \,,
\]
(see (6.11) in Mammen and Park~2005), we establish below that PLS and AIC$_T$ are
equivalent as long as $\hat\sigma^2$ is consistent and $\hat r$ is stable.

\begin{prop}\label{prop:pls} The PLS defined by Mammen and Park~(2005) is equivalent to
AIC$_T$ defined
by Studer et al.~(2005).
\end{prop}

A decomposition of $AIC_T$ leads to
\begin{prop}\label{prop:aic}
\begin{eqnarray*}
\lefteqn{AIC_T - \Big(\frac{1}{n\sigma^2}\verr^\prime\verr - 1 \Big)} \\
&=& \frac{1}{n\sigma^2} ||(I-H)\r||^2 + \frac{1}{n\sigma^2}E[||H\verr||^2] +
O_p(\frac{h^2+w^4}{\sqrt{n}}) + O_p\big(\frac{1}{n\sqrt{w^{d-1}h}}\big)
\end{eqnarray*}
\end{prop}

The first term on the right hand side of the decomposition of AIC$_T$ is the mean
squared bias, whereas the second term is the variance of $\hat r_{ladd}$, both
divided by $\sigma^2$.
Thus, smoothing parameter selection
based on AIC-type model selection criteria leads to asymptotically optimal bias variance
compromise. 

Proofs of Propositions~\ref{prop:pls} and \ref{prop:aic} are given in Appendix.

\section{ Numerical performance } \label{sec:numeric}

\subsection{Simulation studies}

We are interested in investigating how the smoothing parameters
are related to performance of the estimators of general regression function in terms of
conditional MISE. 
For general multivariate nonparametric regression problem, there are
limited simulation studies reported in the literature. 
For example, Banks et al.~(2003) reported comparison results of a broad class of
multivariate nonparametric regression techniques. Some additive model simulation studies
can be found in Dette et al.~(2005) and Martins-Filho and Yang~(2006). Here we
focus on
comparison to local linear and additive estimators as a benchmark on either extremes.
Local linear estimator is
optimal for general regression function estimation so the comparison to it allows us to
assess the
behavior for nonadditive regression function estimation. Likewise additive estimator is
used to
study the behavior for additive regression function estimation. Results are based on
Monte-Carlo
approximation of MISE. 

\paragraph{d=2:} A random uniform design on $[-1,1]^2$ and normally distributed residuals 
${\cal N}(0,\sigma^2)$ were assumed with sample
sizes 200, 400, and 1600. Estimators are evaluated at an
equidistant output grid of $21\times 21$ points. For fitting the SBE, we used
{\it SBF2} package of {\it R} developed in conjunction with Studer et al.~(2005), which is
freely available from www.biostat.uzh.ch/research/software/. 

The main factor of consideration in our simulation studies is the
regression function, covering a range of additive and nonadditive functions. 
To illustrate the behavior of the local additive estimator, we first consider the
regression function 
\begin{equation} \label{e:r2}
r(\x) = x_1^2 + x_2^2  + \frac{\alpha}{1-\alpha} x_1x_2 \,,
\end{equation}
where $\alpha$ controls the amount of nonadditive structure in the function. 
\[
\centering\fbox{Figure~\ref{fig:d2mise0} about here.}
\]

Performance of the local additive estimator is illustrated in Figure~\ref{fig:d2mise0}.
Estimation is based on 400 observations with $\alpha=0.4$ and $\sigma=0.5$.
All estimators used their MISE-optimal smoothing parameters. 
As expected, the additive estimator (lower right panel) does not capture the nonadditive
structure. 
The local linear estimator (upper right) reveals the diagonal structure but has a quite
large bias due to its large MISE-optimal bandwidth ($h=0.64$). 
Because of local additive, instead of local linear, approximation of the regression
function, the local additive estimator uses more observations ($w=0.94$), resulting in an
improved variance, whereas the bandwidth is smaller ($h=0.47$), resulting in an improved
bias. 
As a consequence, the local additive estimator inherits the optimal properties in a {\it
local} sense. 
For smoothing parameter selection in practice, Figure~\ref{fig:d2aic} presents
comparison of ASE-optimal parameters to AIC$_C$ optimal ones for the local additive
estimator based on one realization drawn from the same design used
in Figure~\ref{fig:d2mise0} with $\sigma=0.5$. The range of smoothing parameters
suggested by both criteria largely agrees and we find AIC$_C$ comparable for practical
use.

\[
\centering\fbox{Figure~\ref{fig:d2mise2} about here.}
\]

The effect of nonadditivity $\alpha$ in the regression function on MISE can be seen in
Figure~\ref{fig:d2mise2}, on a log scale. 
MISE~(first row) is decomposed into the integrated squared bias~(second row) and variance~(third row). 
Different columns correspond to different $\sigma$s. 
In each panel, the optimal MISE is plotted as a function of $\alpha$, with an individual optimal choice of smoothing parameters found in the above simulations.
Solid line is for local additive estimator, dashed line for local linear estimator and
dotted line for additive estimator. 
As is expected, the regression function has little effect on the local linear estimator
but had a dramatic impact on the additive estimator because of growing nonadditivity.
Local additive estimator shows relatively robust performance, adapting the best of the
former estimators.

MISE behavior for other regression functions is summarized in Table \ref{tab:mise}.
Regression functions used are
additive peaks
\[
r(\x)=\frac{1}{2}\sum_{k=1}^2
\left(0.3\exp(-2(x_k+0.5)^2)+0.7\exp(-4(x_k-0.5)^2)+0.5\exp(-\frac{x_k^2}{2})\right),
\]
superposed peaks
\[
r(\x)=0.3\exp(-2\|\x+0.5\|^2)+0.7\exp(-4\|\x-0.5\|^2)+0.5\exp(-\frac{
\|\x\|^2}{2})\,,
\]
and periodic nonadditive function
\[r(\x) = \cos(\pi||\x||)\,.\]
MISE-values are multipled by 1000. MISE-optimal smoothing parameters are
also supplied, with MISE ratios.
\[
\centering\fbox{Table~\ref{tab:mise} about here.}
\]
We considered variants of these scenarios for other regression functions and design
densities such as fixed
uniform, fixed uniform jittered, linearly skewed one and linearly skewed
jittered designs and observed similar phenomena stable across designs considered.
More simulation results are found in Park and Seifert~(2008).
There, one can also find simulations for $d=3$. 
Because of dimensionality, the candidate
regions of smoothing parameters are narrower than those for $d=2$, but the behavior of the
estimators is similar and thus the same conclusions apply. 
 
{\bf d=10:} For higher--dimensional case, we considered 
the regression function 
\begin{equation} \label{e:r10}
r(\x) = x_1^2 + \alpha x_1 \big(\sum_{j=2}^{10} x_j) \,\quad \alpha = 0, 0.5, \mbox{ or } 1 \,,
\end{equation}
with 2000 observations on a random uniform design and $\sigma=0.2$. 
Local estimation in 10 dimensions calls for boundary correction. Otherwise, the expected number of observations in a corner would be $w^{10}n$ compared to $1024w^{10}n$ in the center. 
To illustrate the behavior of local additive estimator using an additive
estimator other than the SBE, we used the function {\it gam} in the {\it mgcv} package of
{\it R}.
Although optimality of the penalized splines used there is not known, the idea of
local additive estimator can be easily applied. 
Moreover, {\it gam} has computational advantages; implementation with {\it gam}
particularly facilitates selection of smoothing parameter using generalized cross
validation~(GCV).
Unconditional MASE was approximated with 20 runs of simulation. To reduce computational
burden, estimators are evaluated at 50 design points randomly
chosen at each simulation.
The resulting relative standard error of MASE estimators is about 3-5\%.
\[
\centering\fbox{Figure~\ref{fig:d10mise} about here.}
\]
Figure~\ref{fig:d10mise} shows performance of estimators for three different values of
$\alpha$. 
Dashed line is for local linear estimator, solid line for local additive estimator. The
letter ``a" at the end of solid line represents additive estimator. 
The $x$-axis represents smoothing parameter; for local linear estimator, it is
the
bandwidth $h$ and for local additive estimator, it is $w$, and the GCV-optimal value of
$h$ given $w$ was chosen internally by {\it gam}. 
Performance of local linear estimator does not depend on the regression function,
while local additive estimator adapts to additivity, exhibiting lower curves as the panel
moves to the right. 
We can conclude that overall performance of local additive estimator exceeds that of
others, adapting to
nonadditivity.

In summary, we have observed that when the
regression function is additive or close to additive the local additive estimator is
compatible to the additive estimator, and when the regression function is nonadditive it
mimics the local linear estimator
whenever possible. 
We also have noticed that the lowest possible bandwidth that local additive estimator
could exploit is limited by the number of observations required to obtain a stable
estimator for every output point.
A boundary correction sometimes helps to stabilize an estimator but it works differently
for different estimators and thus we decided not to include it except for $d=10$. 

\subsection{Real data example}

We use the ozone dataset from the {\it R} package (Section 10.3, Hastie and
Tibshirani~(1990)) to make comparison to the previous analysis. With nine predictors, an
additive regression model would be a natural choice. When a new approach which can deal
with nonadditive structure is applied, the model can be further refined or simplified.
Studer et al.~(2005) pointed out that the additive model with nine predictors is almost
equivalent, in terms of adjusted R$^2$, to an additive model with a subset of
predictors, allowing bivariate interaction terms. They applied penalized regression
approach to uncover behavior of the bivariate interaction, noting serious departure from
additive model assumption.

To make it comparable, we adopt the same framework as Studer et al.~(2005), where the
dependent variable is defined as the logarithm of the upland ozone concentration (up03)
and three predictors, humidity (hmdt), inversion base height (ibtp), and calendar
day (day) are chosen which maximize adjusted R$^2$ among fitted additive models with
bivariate interaction terms with 16 degrees of freedom each, using {\it gam} in
{\it R} package {\it mgcv}. Then the three
variables were scaled to [0,1]. As noted in the previous analysis, one observation (92)
that contains excessive value of wind speed was removed prior to the analysis. 

We consider local additive model and additive with bivariate interaction model for
comparison.
The additive with interaction model was fitted using {\it gam} with
internally chosen optimal smoothing parameters. To fit the local additive model based
on the SBE, univariate bandwidths $h_1, h_2$ and $h_3$
are initially chosen to have four degrees of freedom each as in Studer et
al.~(2005). These are shown to lie close together with mean $h
=\sqrt[3]{h_1h_2h_3}=0.237$. Bandwidths for the local additive estimator are set to be
$(ch_1,ch_2,ch_3)$. Parameters $c$ and $w$ are then selected based on
AIC$_C$.

\[
\centering\fbox{Figure~\ref{fig:ozone} about here.}
\]

These estimators are compared in Figure~\ref{fig:ozone}. For
reference, we also reproduced the local linear estimator from Studer et al.~(2005). 
The univariate components on the top show
similar trend, although the local linear estimates show occasional kinks and the
additive
with interaction models tends to smooth out quickly, especially for hmdt. The bottom
row shows the largest bivariate interaction, that between ibtp and hmdt, for each
estimator. We see that in both terms the local additive estimator provides a good
compromise.
Interested readers are referred to Section 5.3 and Figure 5 in Studer et al.~(2005)
for
further comparison and issues with regularisation. 

\begin{table}
\begin{center}
\begin{tabular}{|c|p{1.6in}|p{2.3in}|p{1.75in}|} \hline
\multicolumn{4}{|c|}{\bf Additive peaks} \\ \hline\hline
$\sigma$ & local linear ($h_{opt}$) & local additive ($h_{opt},w_{opt}$) & additive
($h_{opt}$) \\ \hline 
0.1&3.9=315\% (h=0.260)&1.2=100\% (h=0.123, w=0.870)&1.3=107\% (h=0.143) \\ \hline 
0.5&22.1=136\% (h=0.473)&16.2=100\% (h=0.350, w=0.988)&15.4=95\% (h=0.350)\\ \hline 
1.0&39.6=111\% (h=1.000)&35.6=100\% (h=0.741, w=0.933)&32.5=91\% (h=0.861)\\ \hline
\end{tabular}\\
\begin{tabular}{|c|p{1.6in}|p{2.3in}|p{1.75in}|} \hline
\multicolumn{4}{|c|}{\bf Superposed peaks} \\ \hline\hline
$\sigma$ & local linear ($h_{opt}$) & local additive ($h_{opt},w_{opt}$) & additive
($h_{opt}$) \\ \hline 
0.1&2.6=117\% (h=0.260)&2.2=100\% (h=0.193, w=0.242)&7.0=311\% (h=0.350)\\ \hline
0.5&14.9=124\% (h=0.741)&12.0=100\% (h=0.638, w=0.716) &13.3=110\% (h=0.638)\\ \hline 
1.0&30.9=123\% (h=1.000)&25.1=100\% (h=0.741, w=0.741) &24.4=97\% (h=0.861)\\ \hline
\end{tabular}\\
\begin{tabular}{|c|p{1.6in}|p{2.3in}|p{1.75in}|} \hline
\multicolumn{4}{|c|}{\bf Periodic nonadditive} \\ \hline\hline
$\sigma$ & local linear ($h_{opt}$) & local additive ($h_{opt},w_{opt}$) & additive
($h_{opt}$) \\ \hline 
0.1&4.8=130\% (h=0.260)&3.7=100\% (h=0.193, w=0.242)&96.8=2611\% (h=0.166)\\ \hline 
0.5&32.7=97\% (h=0.350)&33.6=100\% (h=0.260, w=0.260)&111.7=333\% (h=0.302)\\ \hline 
1.0&85.7=91\% (h=0.473)&93.9=100\% (h=0.407, w=0.407)&139.2=148\% (h=0.473)\\ \hline
\end{tabular}\\
\caption{Comparison of MISE performance based on 400 observations at different
standard
deviations--optimal parameters are given in the parentheses. Outperformance of
local additive estimator is consequence of smaller $h$ than that for local linear
estimator and smaller additive region ($w<1$) than that for additive estimator.} 
\label{tab:mise}
\end{center}
\end{table}

\begin{figure} [htp]
\includegraphics[width=1.0\textwidth]{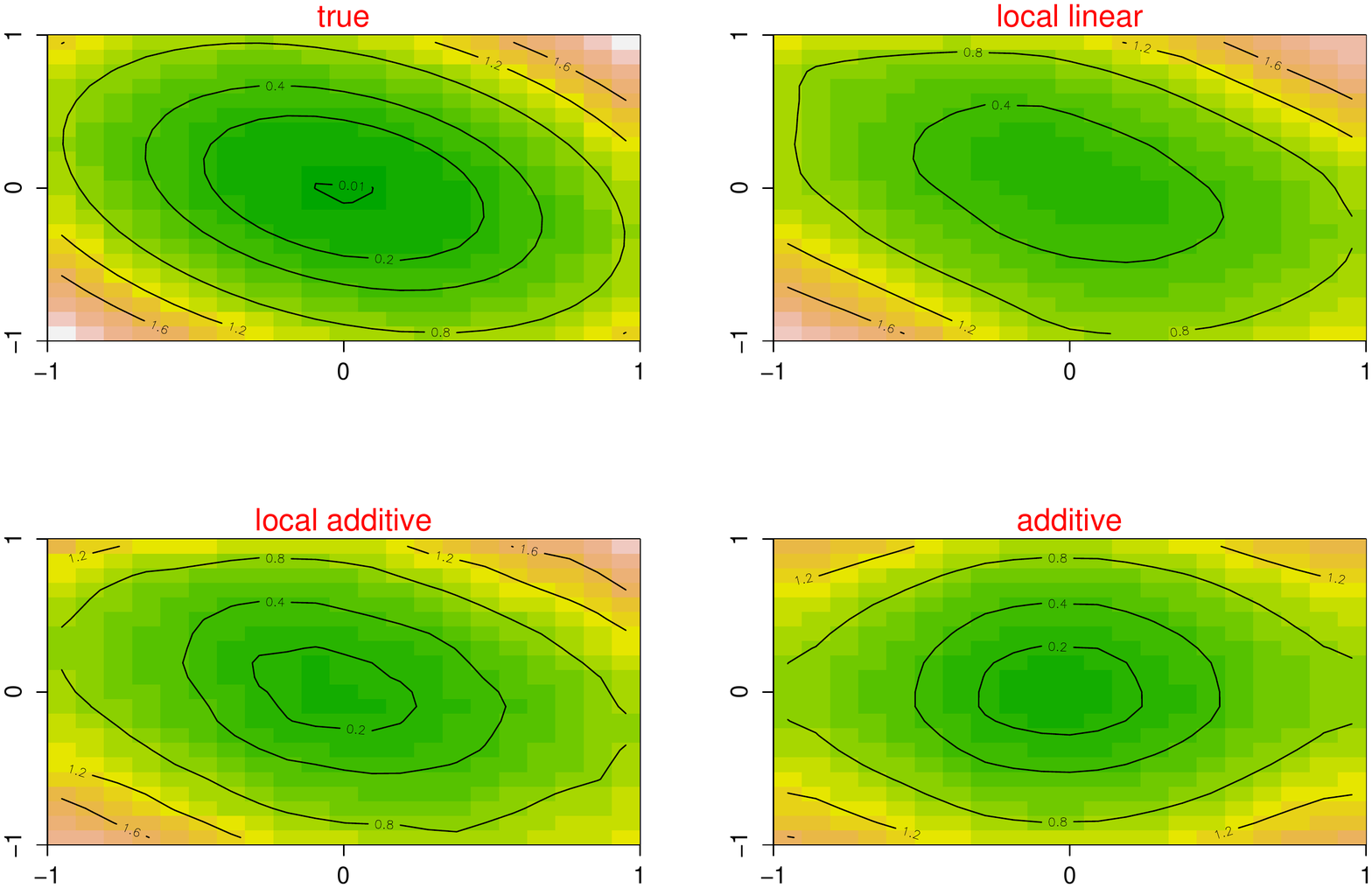} 
\caption{Contour plot of regression function~(\ref{e:r2}) and estimators. Parameters
are chosen to be MISE-optimal from simulation with $\alpha=0.4$ and $\sigma=0.5$. Additive
estimator fails to capture nonadditive structure.
While local linear estimator and local additive estimator show compatible performance,
local additive estimator incurs smaller bias at the center due to smaller bandwidth.}
\label{fig:d2mise0}
\end{figure}

\begin{figure} [htp]
\includegraphics[width=1.0\textwidth]{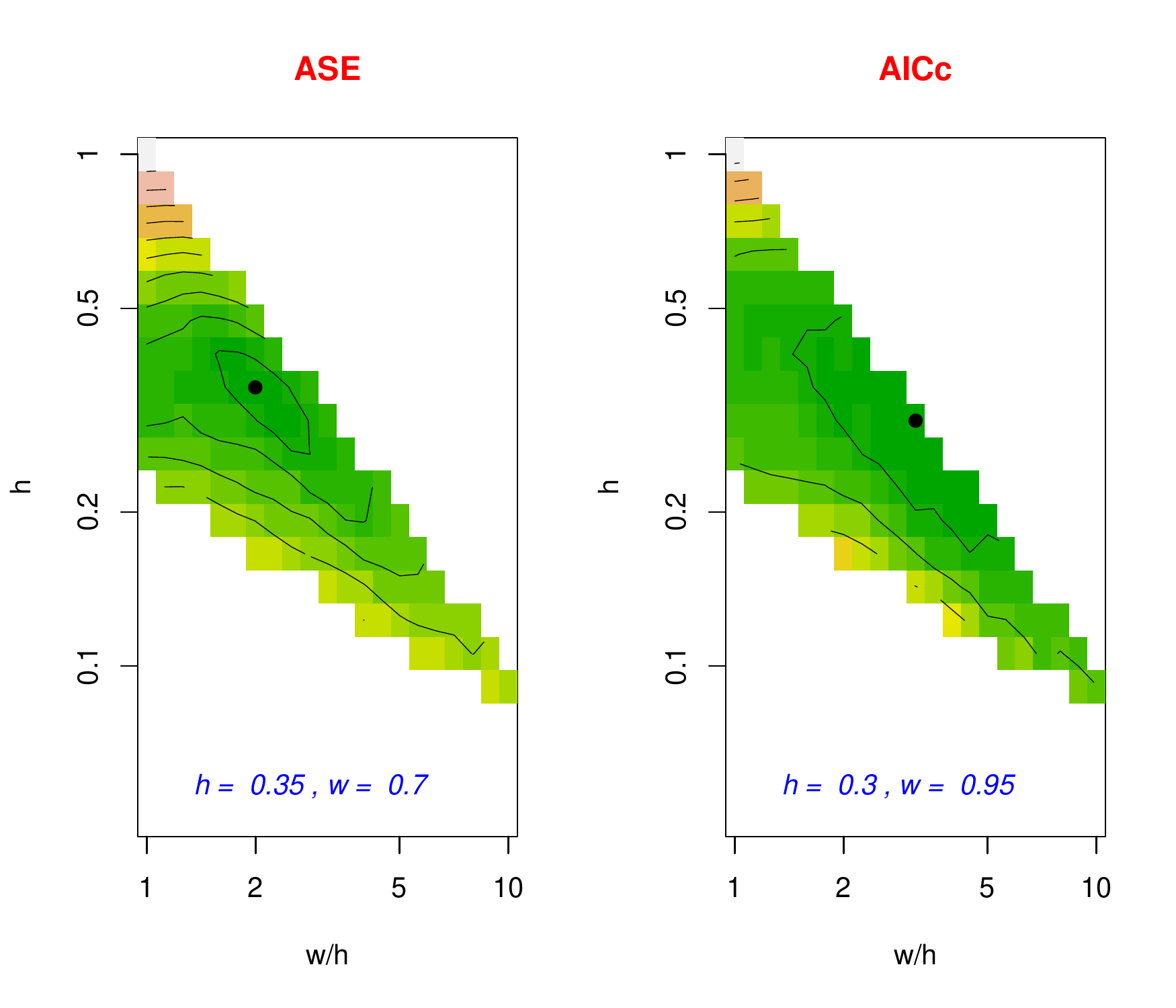} 
\caption{ASE and parameter selection by AIC$_{C}$ for regression function~(\ref{e:r2}) and
design used in Figure~\ref{fig:d2mise0} with $\sigma=0.5$.}
\label{fig:d2aic}
\end{figure}

\begin{figure} [htp]
\includegraphics[width=1.0\textwidth]{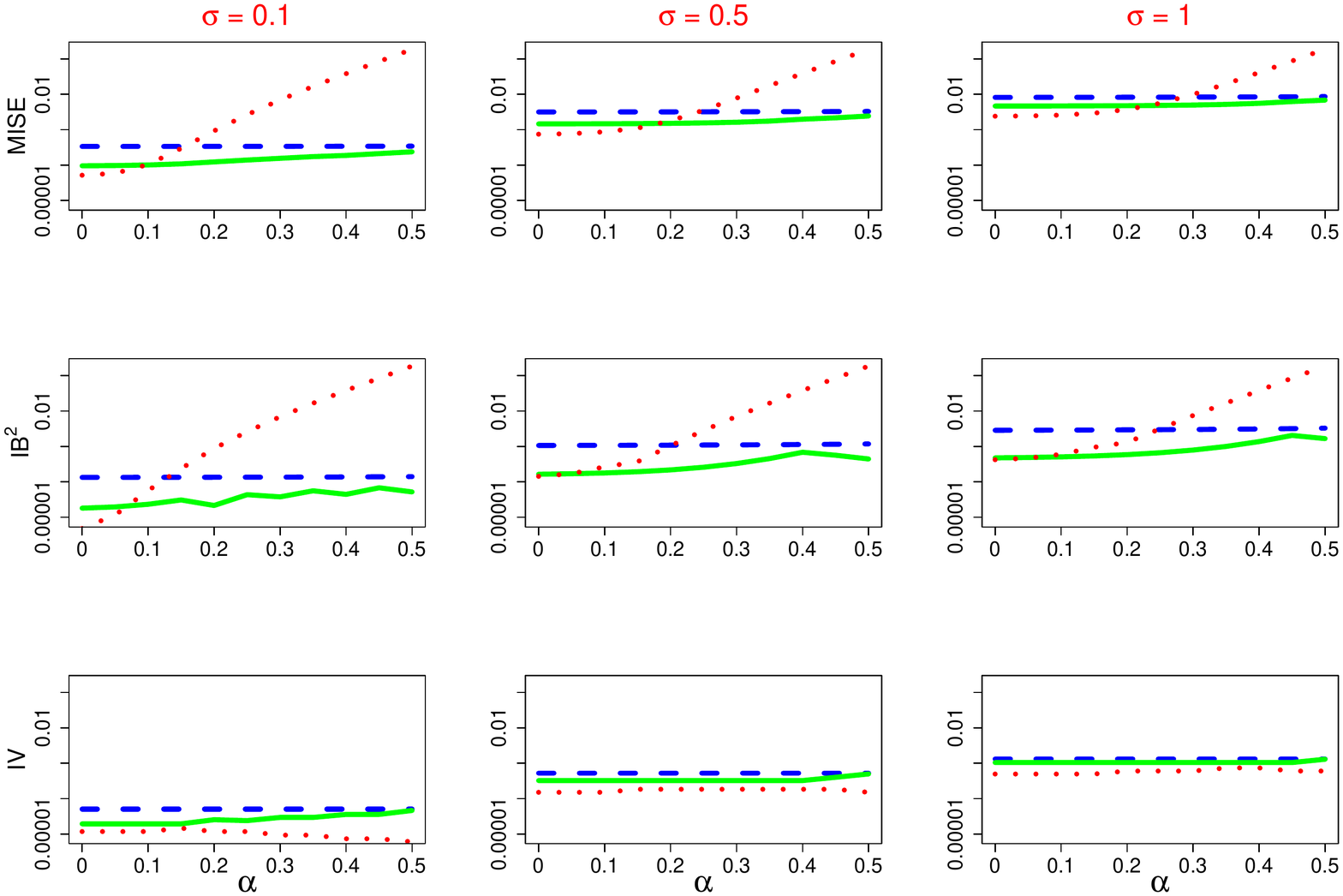}
\caption{Effect of nonadditive regression function on the MISE performance for
MISE-optimal parameters. MISE (first row), integrated squared bias (second row) and
variance (third row) as functions of $\alpha$ in (\ref{e:r2}) is plotted on a log scale
for increasing $\sigma$. Local linear estimator (dashed line) is not affected but additive
estimator (dotted line) dramatically deteriorates. Local additive estimator (solid line)
shows relatively robust performance.} \label{fig:d2mise2}
\end{figure}


\begin{figure} [htp]
\includegraphics[width=1.0\textwidth]{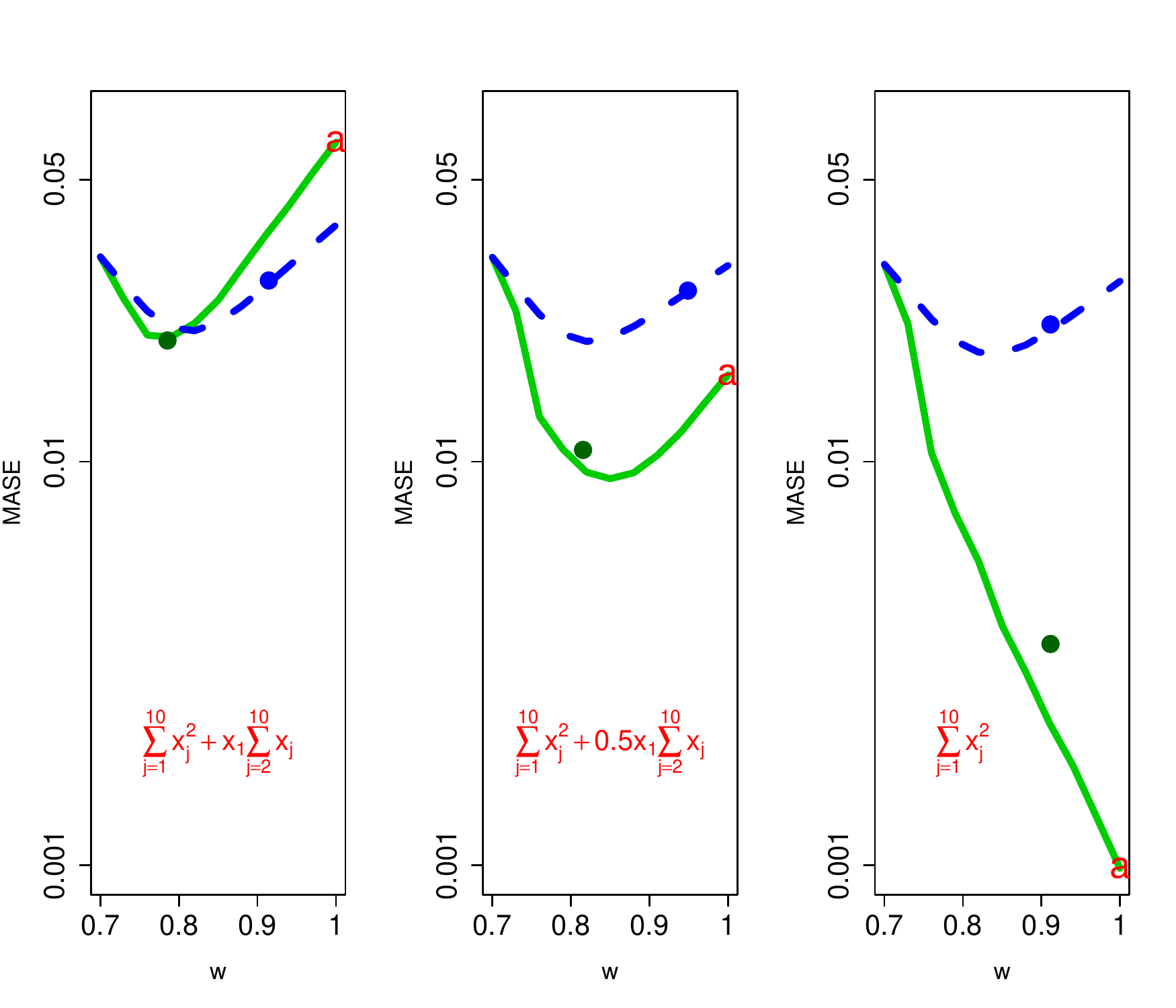} 
\caption{Comparison of unconditional MASE performance of a 10--dimensional regression
function~(\ref{e:r10}) for local linear estimator (-\,-), local additive estimator (--)
and additive estimator (``a"). $x$-axis represents bandwidths for local linear
estimator and
$w$ for local additive estimator with an internal choice by {\it gam} of $h$ at given $w$.
Dots and ``a" show mean MASE at GCV-optimal smoothing parameters vs.  mean
GCV-optimal $w$.}
\label{fig:d10mise}
\end{figure} 

\begin{figure} [htp]
\includegraphics[width=1.0\textwidth]{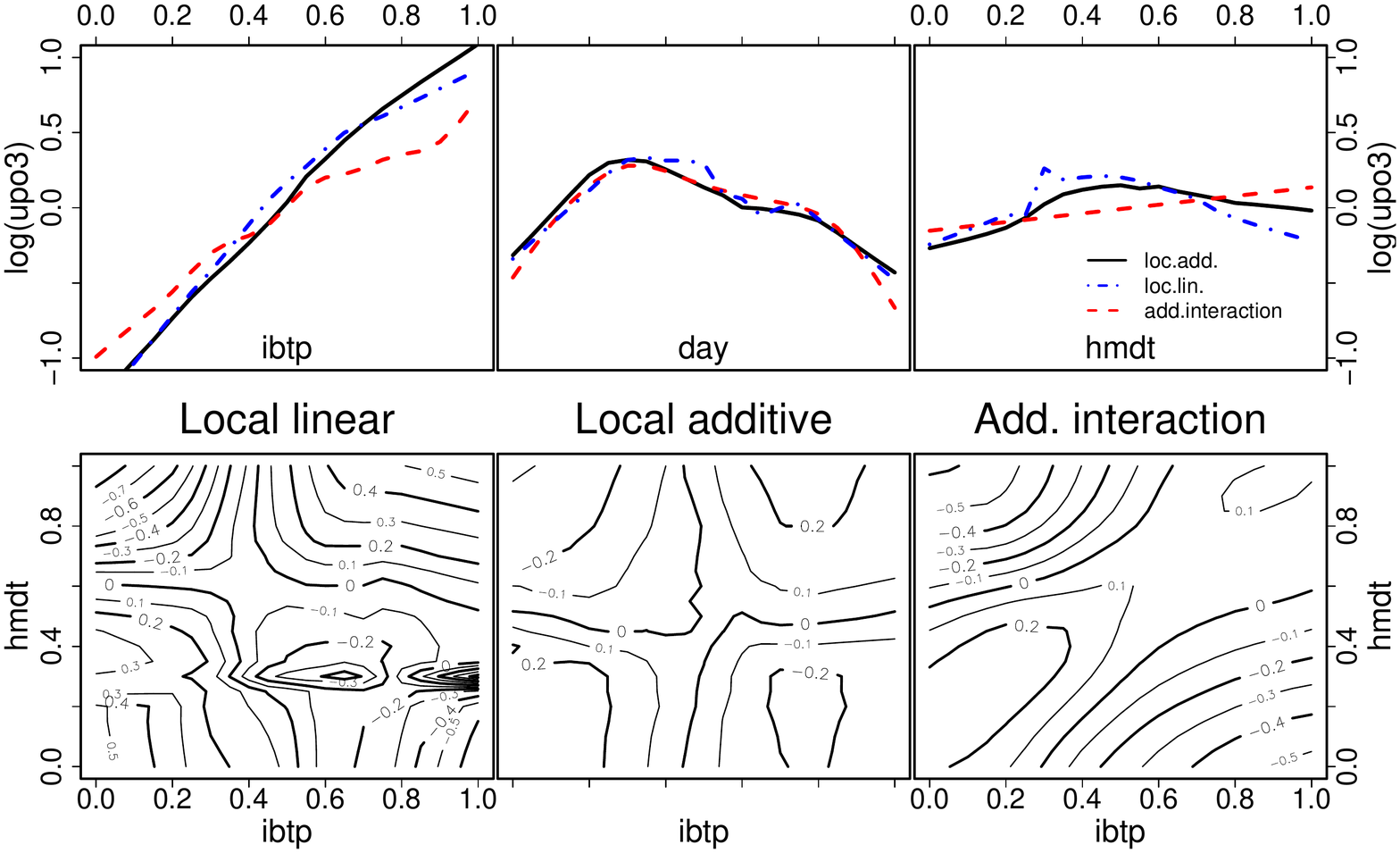} 
\caption{Comparison of local linear, local additive, and additive spline with
interaction estimators. Top row
shows univariate additive components and bottom row shows bivariate components of
ibtp and hmdt for each estimator.}
\label{fig:ozone}
\end{figure}

\section*{Appendix}

\paragraph{Proof of Proposition~\ref{prop:pls}} We use the following fact for additive
regression functions
\[
tr(H) \rt K(0)\sum_{j}1/h_j:=tr(H)_{\infty} \,,
\]
which can be deduced from (6.11a) or (6.11) in Mammen and Park~(2005).
Thus, PLS (p.~1269, Mammen and Park~2005) is defined for additive regression functions as
\[
PLS = \hat\sigma^2\left(1+2\frac{tr(H)_{\infty}}{n}\right) \,.
\]
In this form, it can be generalized to nonadditive functions. Firstly from the definition
of AIC$_T$, it can be written as
\[
AIC_T + 1 = \frac{1}{\sigma^2}\left(PLS +
2(\sigma^2-\hat\sigma^2)\frac{tr(H)}{n}+2\hat\sigma^2\frac{tr(H)-tr(H)_{\infty}}{n}
\right) \,.
\]
Then observe that
\[
2(\sigma^2-\hat\sigma^2)\frac{tr(H)}{n}+2\hat\sigma^2\frac{tr(H)-tr(H)_{\infty}}{n} =
o\left(\frac{tr(H)}{n}\right) \,.
\]
Therefore, it follows that
\[
AIC_T + 1 = \frac{1}{\sigma^2}\left(PLS + o\left(\frac{tr(H)}{n}\right)\right) \,,
\]
as long as $\hat\sigma^2$ is consistent and $\hat r$ is stable. \hfill$\Box$

\paragraph{Proof of Proposition~\ref{prop:aic}} 
AIC$_T$ can be written as
\begin{eqnarray*}
\mbox{AIC}_T &=& \frac{1}{n\sigma^2}\r^\prime(I-H)^\prime(I-H)\r +
\frac{1}{n\sigma^2}\verr^\prime(I-H)^\prime(I-H)\verr \\
& & +\frac{1}{n\sigma^2} 2\verr^\prime(I-H)^\prime(I-H)\r - 1 + \frac{2 tr(H)}{n} \,.
\end{eqnarray*}
Observe that
\[
\verr^\prime(I-H)^\prime(I-H)\verr = \verr^\prime\verr - 2E[\verr^\prime H\verr] +
O_p(\sqrt{V[\verr^\prime H\verr]}) + E[\verr^\prime H^\prime H\verr]
+ O_p(\sqrt{V[\verr^\prime H^\prime H\verr]})  \,.
\]
Since $E[\verr^\prime H\verr] = tr(H)\sigma^2$, we have
\begin{eqnarray*}
\lefteqn{\mbox{AIC}_T - \Big(\frac{1}{n\sigma^2}\verr^\prime\verr - 1\Big)} \\
&=& \frac{1}{n\sigma^2}||(I-H)\r||^2 + \frac{1}{n\sigma^2}E[\verr^\prime H^\prime H\verr]
+\frac{1}{n\sigma^2} 2\verr^\prime(I-H)^\prime(I-H)\r \\
& & + \frac{1}{n\sigma^2}O_p(\sqrt{V[\verr^\prime H\verr]}) +
\frac{1}{n\sigma^2}O_p(\sqrt{V[\verr^\prime H^\prime H\verr]})\big) \,.
\end{eqnarray*}

The rest follows from a series of lemmas below.

\begin{lem}
\[
tr((H^\prime H)^\prime(H^\prime H))=O\big(tr(H^\prime H)\big)=O\big(1/(w^{d-1}h)\big)
\]
\end{lem}

{\bf Proof:} Denote by $H_i$ the hat matrix of the additive estimator used for local
additive estimation at $x_0=X_i$.
Then, inflating the matrix to an $n\times n$ matrix, the $i$th line of $H$ is the $i$th
line of $H_i$ := $H_{i,i}$\,.
Now, considering the form of the estimator $\hat{r}_i = \hat{r}_0+ \hat{r}_1+ \ldots +
\hat{r}_d$, where all components are oracle, we have 
\[
H_{i,j} = \left\{\begin{array}{cc}
O(\frac{1}{\tilde n\tilde h}) & \text{ if for all } k:|X_{i_k}-X_{j_k}|\leq w \text{ and
for some } k: |X_{i_k}-X_{j_k}|\leq h \\
O(\frac{1}{\tilde n}) & \text{ if for all } k: |X_{i_k}-X_{j_k}|\leq w \text{ and for 
all } k: |X_{i_k}-X_{j_k}|> h\\
0 & \text{ otherwise}\end{array}\right.
\]
Note, that these $O()$s are uniform over $X$ because of (A.5), using Gao (2003) as in
Studer et al. (2005).
Let's first look at $tr(H'H)$.
We have 
\[
H_{i,i}' H_{i,i} = O\left(O(\frac{1}{\tilde n\tilde h})^2  O(\tilde n\tilde h) +
O(\frac{ 1}{ \tilde n})^2 O(\tilde n)\right) = O\left(\frac{1}{\tilde n\tilde h} \right) =
O\left(\frac{1}{nw^{d-1}h} \right). 
\]
Consequently, 
\[ 
tr(H'H) =  O\left(\frac{1}{w^{d-1}h} \right). 
\]

Now, look at the general elements of $H'H$.
With a slight abuse of notation, 
\begin{eqnarray*} 
H_{i,i}'H_{j,j} &=& O\left(O(\frac{1}{\tilde n\tilde h})^2 
O(\tilde n\tilde h) + O(\frac{1}{\tilde n\tilde h})O(\frac{ 1}{ \tilde n})  O(\tilde
n\tilde h) + O(\frac{ 1}{ \tilde n})^2 O(\tilde n)\right) \\
& &\qquad\text{ if for all } k: (X_{i_k}\pm w)\cap (X_{j_k}\pm w) \not= \varnothing \\
 & &\qquad\text{ and for some } k:   (X_{i_k}\pm h)\cap (X_{j_k}\pm h)\not= \varnothing 
\\
& & O\left(O(\frac{1}{\tilde n\tilde h})O(\frac{ 1}{ \tilde n})  O(\tilde n\tilde h) +
O(\frac{ 1}{ \tilde n})^2 O(\tilde n)\right)\\
& & \qquad\text{ if for all } k:(X_{i_k}\pm w)\cap (X_{j_k}\pm w) \not= \varnothing \\
 & & \qquad\text{ and for some } k:   (X_{i_k}\pm w)\cap (X_{j_k}\pm h)\not= \varnothing 
\\
 & & \qquad\text{ or }   (X_{i_k}\pm h)\cap (X_{j_k}\pm w)\not= \varnothing  \\
 & & O\left(O(\frac{ 1}{ \tilde n})\right) \qquad\text{ if for all } k: (X_{i_k}\pm
w)\cap (X_{j_k}\pm w) \not= \varnothing \\
& & \qquad \text{ and } (X_{i_k}\pm w)\cap (X_{j_k}\pm h)= \varnothing \text{ and }
(X_{i_k}\pm h)\cap (X_{j_k}\pm w)= \varnothing  \\
 & & 0 \quad\text{ otherwise} \,.
\end{eqnarray*} 
Finally,
\begin{eqnarray*} 
H_{i,i}'H_{j,j} &=& O(\frac{1}{\tilde n\tilde h})   \text{ if for all } k: (X_{i_k}\pm
w)\cap (X_{j_k}\pm w) \not= \varnothing \\
 & & \qquad \text{ and for some } k:   (X_{i_k}\pm h)\cap (X_{j_k}\pm h)\not= \varnothing 
\\
 & & O(\frac{ 1}{ \tilde n})  \text{ if for all } k: (X_{i_k}\pm w)\cap (X_{j_k}\pm w)
\not= \varnothing  \\
& & \qquad \text{ and for all } k:   (X_{i_k}\pm h)\cap (X_{j_k}\pm h)= \varnothing \\
 & & 0 \text{ otherwise}\,.
\end{eqnarray*} 
Thus, $H'H$ has the same structure as $H$, of course with different constants and larger
non-zero regions, but all of the same order. Therefore, 
\[
tr(H'HH'H) = O(tr(H'H)) = O\left(\frac{1}{w^{d-1}h} \right). 
\]

\begin{lem}
\[
\frac{1}{n\sigma^2}\verr^\prime(I-H)^\prime(I-H)\r = \frac{1}{n\sigma^2}<(I-H)\verr,
(I-H)\r> = O_p\big(\frac{h^2+w^4}{\sqrt{n}}\big)
\]
\end{lem}

{\bf Proof:} First note that $(I-H)\r  = O((h^2+w^4)\one)$. 
\begin{eqnarray*}
\lefteqn{\frac{1}{n\sigma^2}<(I-H)\verr, (I-H)\r>} \\ 
&\leq& \frac{1}{n\sigma^2}||(I-H)\verr||||(I-H)\r|| \\
&\leq& \frac{1}{n\sigma^2}O_p(\sqrt{n})O((h^2+w^4)\one) =
O_p\big(\frac{h^2+w^4}{\sqrt{n}}\big)
\end{eqnarray*}
where the last inequality follows from
\[
||(I-H)\verr||^2 = O_p(tr(V[||(I-H)\verr||^2)) = O_p(tr((I-H)(I-H)^\prime\sigma^2) =
O_p(n) \,.
\]

\begin{lem}
\[
V[\verr^\prime H\verr] = V[<\verr, H\verr>] = O(E[||H\verr||^2]) 
\]
If $tr((H^\prime H)(H^\prime H))=O(tr(H^\prime H)$, then 
\[
V[\verr^\prime H^\prime H\verr] = V[||H\verr||^2]= O(E[||H\verr||^2])
\]
\end{lem}

{\bf Proof:} Similar to the proof of lemma 5 in the appendix of Studer et al. (2005), we
use the following fact: For symmetric matrices $B$ and $C$ and $E[\verr^4] =
(3+\kappa)\sigma^4$,
\[
Cov(\verr^\prime B\verr, \verr^\prime C\verr) = 2\sigma^4tr(BC) + \kappa\sigma^4tr(B\cdot 
diag(C))
\]
Putting $B=C=\frac{1}{2}(H+H^\prime)$ gives
\begin{eqnarray*}
V[\frac{1}{2}\verr^\prime(H+H^\prime)\verr] &=& 
2\sigma^4(\frac{1}{4}tr(HH+2H^\prime H+H^\prime H^\prime) + \kappa\sigma^4
tr(\frac{1}{4}(H^\prime+H) diag(H+H^\prime)) \\
V[\verr^\prime H\verr] &=& \sigma^4(tr(HH + H^\prime H) + \kappa\,tr(diag(H)^2) \\
&\leq& \sigma^4(tr(HH)+tr(H^\prime H) + |\kappa| tr(H^\prime H)
\end{eqnarray*}
Using the equivalence of the trace to Hilbert-Schmidt norm,
\[
||H||_{HS}^2 = tr(H^\prime H)
\]
it follows that
\[
tr(HH) = <H^\prime, H>_{HS} \leq ||H^\prime||_{HS}||H||_{HS} = ||H||_{HS}^2 =
tr(H^\prime H) \,.
\]
Hence,
\begin{eqnarray*}
V[\verr^\prime H\verr] &\leq& \sigma^4(2tr(H^\prime H)+|\kappa| tr(H^\prime H) ) \\
&=& \sigma^4(2+|\kappa|) tr(H^\prime H) \\
&=& \sigma^2(2+|\kappa|) E[||H\verr||^2]
\end{eqnarray*}
Thus, $V[<\verr, H\verr>] = O(E[||H\verr||^2])$. 
Moreover, replacing $H$ by $H^\prime H$ in the above leads to
\begin{eqnarray*}
V[\verr^\prime H^\prime H\verr] &\leq& \sigma^2(2+|\kappa|)
tr((H^\prime H)^\prime(H^\prime H)) \,.
\end{eqnarray*} 

Hence, if $tr((H^\prime H)^\prime(H^\prime H))=O(tr(H^\prime H))$, then $V[||H\verr||^2]
= O(tr(H^\prime H)) = O(E[||H\verr||^2])$. \hfill$\Box$

\noprint{
\section*{Appendix}

\subsection{ Asymptotic properties of the SBE } \label{sec:add}

The normal equation derived in the previous section will serve as a basis of our further development to study properties of the estimator. Suppose that data are observed with random error $\varepsilon$. 
If we can show that the operator in the left hand side converges with some (unknown) order, such that
\[
(\P_{add}\S_*\P_{add})^{-1} = (\S_{add,\infty})^{-1}(1+o(1))\,,
\]
then the properties of the estimator can be studied from
\[
\hat\r_{add} = \S_{add,\infty}^{-1}\P_{add}\r_{L}(1+o(1)) \,.
\]
Here $\P_{add}\r_{L}$ has relatively simple formulas to work with.
Indeed, the limiting operator exists for the additive estimator, as is shown in Studer et al.~(2005) and Mammen et al.~(1999).

This knowledge of the limiting operator will be useful to study properties of the estimator for general regression function, as they may be written as
\begin{eqnarray*}
B[\hat \r_{add}(\x)] &=& (\S_{add,\infty}^{-1}E[\P_{add}\r_L(\x)]- \r(\x))(1+o_p(1)) \\
V[\hat \r_{add}(\x)] &=& (\S_{add,\infty}^{-1})^2V[\P_{add}\r_L(\x)](1+o_p(1)) \,.
\end{eqnarray*}

Coversion to the local additive estimation gives rise to a complication with design density function. In the transformed space, the density function is not static, but a function of the local neighborhood $w$. The convergence therefore should account for the behaviors of the parameter $w$ as well as the smoothing parameter $h$ as $n$ increases. An extended result for local additive estimator is given in Lemma~\ref{lem:sadd} in Section~\ref{sec:proof}.

As we shall see in the subsequent sections, the intricate machinery with respect to the dynamic density function actually helps, in our case, remove the global dependence of additive estimator. To briefly elaborate the idea, let us first introduce some more notations. Put $\mu_j(K^l) = \int u^jK^l(u)\,du$. Additionally for a $d$--dimensional function $g$ we write marginal functions as
\[
(g)_j(x_j) = \int g(\x)\,d\x_{-j} \,,\quad j=1,\cdots,d\,.
\]
2, 3--dimensional marginals, $(g)_{j,k}$, $(g)_{j,k,l}$ are defined similarly.
These should not be confused with $(g)_j'(x_j)$ and $(g)_j''(x_j)$, which are derivatives of 1--dimensional marginals.

Then, it can be shown that
\begin{eqnarray} 
(\S_{add,\infty}\r)^0(\x) &=& \frac{1}{2^{d-1}}\sum_{j=1}^d (f)_j(x_j)(\P_{add}\r)_j^0(x_j) + \frac{1}{2^{d-1}}\sum_{j=1}^d \sum_{k\neq j}^d\int (f)_{j,k}(x_j,u)(\P_{add}\r)^0_k(u)\,du \nonumber\\
& & - \frac{(d-1)}{2^d}\sum_{j=1}^d \int (f)_j(u)(\P_{add}\r)^0_j(u)\,du  \label{e:sadd0}\\
(\S_{add,\infty}\r)^j(\x) &=& \frac{1}{2^{d-1}}\mu_2(K)(\P_{add}\r)^j(x_j)(f)_j(x_j) \nonumber\,.
\end{eqnarray}
An alternative formula for the projection operator $\P_{add}\r$ is given in
Lemma~\ref{lem:padd} in Section~\ref{sec:add}.
Observe that the inverse operator involves all $d+1$ components of the estimator, which hinders the study of an additive estimator for general regression function. 
When the true regression function is assumed to be additive, however, simplification of bias calculation is possible. Write
\[
B[\hat\r_{add}(\x)] = \S_{add}^{-1}E[\P_{add}\r_L(\x)-\S_{add}\r(\x)] \,,
\]
and the order of bias follows from
\[
E[\P_{add}\r_L(\x) -\S_{add}\r(\x)]\,.
\]
Simplification of the equation is given in Proposition~\ref{prop:badd}. 
As is explained, we focus on the quantity $\P_{add} r_L$ to derive asymptotic results.
Because of similarity to the additive framework, general derivations will be made based on
an additive estimator first in the following section. Properties of additive estimator for
additive function have been well established in Mammen et al.~(1999). Our formulation is
slightly different but the results agree with them. In addition, some aspects of the
additive estimator regarding nonadditive function will be studied.

Let $\s = 2^{d-1}\P_{add}\r_L$. Recall that $\r_L(\x) = L(\x)$ defined in~(\ref{e:L0}) and
(\ref{e:Lk}). Assume that $K$ is a product kernel. It follows from Lemma~\ref{lem:padd},
\begin{eqnarray*}
s^0(\x) &=& \frac{1}{n}\sum_{i=1}^n Y_i\Big\{\sum_{j=1}^d K_{h_j}(X_{i,j},x_j)-\frac{(d-1)}{2}\Big\} := \frac{1}{n}\sum_{i=1}^n Y_iW_i  \\
s^j(\x) &=& \frac{1}{n}\sum_{i=1}^n Y_i\frac{X_{i,j}-x_j}{h_j}K_{h_j}(X_{i,j},x_j)\,.
\end{eqnarray*}

Since $\S_{add,\infty}$ is a bounded operator, the bias and variance derivation can be based on those of $\s^0(\x)$, which resembles a standard kernel type estimator. As $n\rt\infty$, asymptotic behavior of those conditional bias and variance may be studied in a standard manner. Because we are interested in ``intercept'' of the estimator, we focus on $s^0(\x)$ from now on. 

\begin{prop} \label{prop:vadd}
Let $V_{add}(\x) = 2^{d-1}V[(\P_{add}\r_L)^0(\x)]$.
\[
V_{add}(\x) = \frac{\sigma^2}{n}\mu_0(K^2)\sum_{j=1}^d\frac{1}{h_j}(f)_j(x_j)\big(1+o_p(1)\big) \,.
\]
\end{prop}
The result follows from standard calculation.

Bias calculation depends on whether or not the true regression function is additive. Let $r_{add}=\P_{add}\r$ be the additive projection of the true regression function. Write
\[
\P_{add}\r_L(\x) = \S_{add}\r_{add}(\x) + \Big(\P_{add}\r_L(\x) - \S_{add}\r_{add}(\x)\Big) \,.
\]
Then,
\[
E[\hat\r_{add}(\x)] = \r_{add}(\x) + \S_{add}^{-1}E[(\P_{add}\r_L)(\x) - (\S_{add}\r_{add})(\x)] \,.
\]
Non--additive bias is reflected on the last term.
Define
\[
B_{add}(\x) = 2^{d-1}E[(\P_{add}\r_L - \S_{add}\r_{add})^0(\x)] \,.
\]
This quantity may be used as an indication of nonadditive bias. Its full expression is
summarized in the following.

\begin{prop} \label{prop:badd}
\begin{eqnarray*}
\lefteqn{B_{add}(\x)} \\
&=& \sum_{j=1}^d\int\frac{1}{n}\sum_{i=1}^n K_h(\X_i,\x)\big(r(\x)-r_{add}^0(\x)\big)\,d\x_{-j} \\
& & + \sum_{j,k=1}^d\int\frac{1}{n}\sum_{i=1}^n K_h(\X_i,\x)\frac{X_{i,k}-x_k}{h_k}\big(r_k'(\x)h_k-r_{add}^k(x_k)\big)\,d\x_{-j} \\
& & - \frac{d-1}{2}\int\frac{1}{n}\sum_{i=1}^n K_h(\X_i,\x)\big(r(\x)-r_{add}^0(\x)\big)\,d\x \\
& & -\frac{d-1}{2}\sum_{k=1}^d\int\frac{1}{n}\sum_{i=1}^n K_h(\X_i,\x)\frac{X_{i,k}-x_k}{h_k}\big(r_k'(\x)h_k-r_{add}^k(x_k)\big)\,d\x \\
& & + \frac{1}{2}\sum_{j,k,l=1}^d\int\frac{1}{n}\sum_{i=1}^n K_h(\X_i,\x)(X_{i,k}-x_k)(X_{i,l}-x_l)r_{k,l}''(\x)\,d\x_{-j} \\
& & -\frac{d-1}{4}\sum_{k,l=1}^d\int\frac{1}{n}\sum_{i=1}^n K_h(\X_i,\x)(X_{i,k}-x_k)(X_{i,l}-x_l)r_{k,l}''(\x)\,d\x + o_p(h^2) \,.
\end{eqnarray*}
\end{prop}
\paragraph{ Proof of Proposition~\ref{prop:badd} } Recall $s(\x) = 2^{d-1}\P_{add}\r_L(\x)$.
\[
E[s^0(\x)] = \sum_{j=1}^d\Big(\frac{1}{n}\sum_{i=1}^n K_{h_k}(X_{i,j},x_j)r(\X_i)\Big) - \frac{d-1}{2}\Big(\frac{1}{n}\sum_{i=1}^n r(\X_i)\Big)
\]
Thus, we have
\begin{eqnarray*}
\lefteqn{\frac{1}{n}\sum_{i=1}^n K_{h_k}(X_{i,j},x_j)r(\X_i)} \\
&=& \int \frac{1}{n}\sum_{i=1}^n K_h(\X_i,\x) r(\X_i)\,d\x_{-j} \\
&=& \int \frac{1}{n}\sum_{i=1}^nK_h(\X_i,\x) r(\x)\,d\x_{-j} + 
\int \frac{1}{n}\sum_{i=1}^nK_h(\X_i,\x) \sum_{k=1}^d(X_{i,k}-x_k)r_k'(\x)\,d\x_{-j} \\
& & + \frac{1}{2}\int \frac{1}{n}\sum_{i=1}^nK_h(\X_i,\x) \sum_{k,l=1}^d(X_{i,k}-x_k)(X_{i,l}-x_l)r_{k,l}''(\x)\,d\x_{-j} + o_p(h^2) \,,
\end{eqnarray*}
and
\begin{eqnarray*}
\lefteqn{\frac{1}{n}\sum_{i=1}^n r(\X_i)} \\
&=& \int \frac{1}{n}\sum_{i=1}^n K_h(\X_i,\x) r(\X_i)\,d\x \\
&=& \int \frac{1}{n}\sum_{i=1}^nK_h(\X_i,\x) r(\x)\,d\x + 
\int \frac{1}{n}\sum_{i=1}^nK_h(\X_i,\x) \sum_{k=1}^d(X_{i,k}-x_k)r_k'(\x)\,d\x \\
& & + \frac{1}{2}\int \frac{1}{n}\sum_{i=1}^nK_h(\X_i,\x) \sum_{k,l=1}^d(X_{i,k}-x_k)(X_{i,l}-x_l)r_{k,l}''(\x)\,d\x + o_p(h^2) \,.
\end{eqnarray*}
Hence, we have
\begin{eqnarray*}
\lefteqn{2^{d-1}(\S_{add}\r_{add})^0(\x)} \\
&=& \sum_{j=1}^d\int\frac{1}{n}\sum_{i=1}^n K_h(\X_i,\x) r_{add}^0(\x)\,d\x_{-j}-\frac{d-1}{2}\int\frac{1}{n}\sum_{i=1}^n K_h(\X_i,\x) r_{add}^0(\x)\,d\x \\
& &+ \sum_{j,k=1}^d\int \frac{1}{n}\sum_{i=1}^n K_h(\X_i,\x)(X_{i,k}-x_k) r_{add}^k(x_k)\,d\x_{-j} \\
& & - \frac{d-1}{2}\sum_{k=1}^d\int\frac{1}{n}\sum_{i=1}^n  K_h(\X_i,\x)(X_{i,k}-x_k) r_{add}^k(x_k)\,d\x \,.
\end{eqnarray*}
and the result follows.

As is expected, bias for general regression function is of order $O(1)$. The smaller order
                                                                                          
 comparison terms would not be necessary to show that. But we used a redundant expression
to make compare case of an additive function. When $r$ is additive and $r_{add}$ is
$||\cdot||_*$ projection, $r_{add}(\x)$ is equal to $r(\x)$ and $r_{add}^k(\x)$ to
$h_kr_k'(x_k)$ and thus further simplification is possible.

Combining with~(\ref{e:sadd0}), an explicit derivation of bias for an additive function is possible, as is shown in Mammen et al.~(1999). We state our version to be a reference for the next section.

\begin{cor} 
Suppose that $r$ is additive: $r(\x) = \sum_{j=1}^d r_j(x_j)$. Under the assumptions
(A1)-(A5), it holds that
\[
B[\hat r_{add}(\x)] =  \frac{1}{2}\mu_2(K)\sum_{j=1}^d h_j^2 r_j''(x_j) + o_p(h^2) + O_p((nh^2)^{-\frac{1}{2}})\,.
\]
\end{cor}

We have not enforced a conventional additive representation such as in~(\ref{eq:radd}) with $\int r_j(x_j)(f)_j(x_j)\,dx_j = 0$. If this convention were used, the bias would be expressed in terms of $r_k''(x_k) - \int r_k''(x_k)(f)_k(x_k)\,dx_k$, the same as that appeared in Mammen et al.~(1999).

\bigskip
It has been noted in Theorem~\ref{thm:main} that the behavior of the estimator for
bilinear function is crucial to reveal some aspect of the additive estimator for
nonadditive function. The connection will be made in Section~\ref{sec:proof} after
introducing the local additive estimator properly. 
However, the local additive estimator, by its construction, inherits properties from additive estimator and thus we formulate results here. 

\begin{prop} \label{prop:nadd}
Suppose that the design points are equally spaced and 
\[
b(\x) = (x_j - \bar X_j)(x_k - \bar X_k) \,.
\]
Then
\[
E[\hat b_{add}(\x)] \equiv 0
\]
independently of $h$. 
\end{prop}

\paragraph{ Proof of Proposition~\ref{prop:nadd} } Recall that
\[
\S_{add}\hat\r_{add} = \P_{add}\r_L \,.
\]

\begin{eqnarray}
E[s^0(\x)] &=& \frac{1}{n}\sum_{i=1}^n r(\X_i)W_i \label{e:conbias0}\,, \\
V[s^0(\x)] &=& \frac{\sigma^2}{n^2}\sum_{i=1}^n W_i^2 \label{e:convar0}\,.
\end{eqnarray}
Similarly, for the $j$th component, $j=1,\cdots, d$, we have
\begin{eqnarray}
E[s^j(\x)] &=& \frac{1}{n}\sum_{i=1}^n r(\X_i)\frac{X_{i,j}-x_j}{h_j}K_{h_j}(X_{i,j},x_j) \label{e:conbiask}\\
V[s^j(\x)] &=& \frac{\sigma^2}{n^2}\sum_{i=1}^n \Big(\frac{X_{i,j}-x_j}{h_j}\Big)^2K_{h_j}^2(X_{i,j},x_j) \label{e:convark}\,.
\end{eqnarray}

Because of symmetry of design points, we can show from~(\ref{e:conbias0}) and ~(\ref{e:conbiask}) that
\[
E[(\P_{add}r_L)(\x)] \equiv 0 \,.
\]

\begin{cor} \label{cor:produnif}
Suppose that the design points are equally spaced and
\[
b(\x) = (x_j-x_j^0)(x_k-x_k^0) \,,
\]
where $\x^0 = (x_j^0, x_k^0)$ is a fixed constant. Then
\[
B[\hat b_{add}(\x)] =  - (x_j-\bar X_j)(x_k-\bar X_k) \,.
\]
In particular, $B[\hat r_{add}(\bar\X)] = 0$.
\end{cor}

\paragraph{ Proof of Corollary~\ref{cor:produnif}} Observe that
\begin{eqnarray*}
r(\x) &=& (x_j-\bar X_j)(x_k-\bar X_k) + (\bar X_j-x_j^0)(x_k-\bar X_k) \\
& & + (x_j-\bar X_j)(\bar X_k-x_k^0) + (\bar X_j-x_j^0)(\bar X_k-x_k^0) \,.
\end{eqnarray*}
Because the last three terms are linear, the additive estimator is unbiased.  
Thus, it may be written as
\begin{eqnarray*}
E[\hat r_{add}(\x)] &=& 0 +(\bar X_j-x_j^0)(x_k-\bar X_k)+ (x_j-\bar X_j)(\bar X_k-x_k^0) \\
& &  + (\bar X_j-x_j^0)(\bar X_k-x_k^0) \\
&=&function- (x_j-\bar X_j)(x_k-\bar X_k) \,.
\end{eqnarray*}
As a result, bias of product function may be written as
\[
B[\hat r_{add}(\x)] = -(x_j-\bar X_j)(x_k-\bar X_k) \,,
\]
and especially when $\x = \bar\X$, it holds that
\[
B[\hat r_{add}(\bar\X)] = 0 \,.
\]

Observe that bias does not depend on how the bilinear function is centered.
When the design points are not equally spaced, the relation holds asymptotically. 

\begin{cor}\label{cor:addprod}
Suppose
\[
b(\x) = x_jx_k \,.
\]
If two- and three--dimensional marginals of the design variable $\X$ are independent, then we have
\[
E[\hat b_{add}(\x)] \rt 0 \,,
\]
and 
\[
B[\hat b_{add}(\x)] = -(x_j-\bar X_j)(x_k-\bar X_k) + o_p(1) \,.
\]
If two- and three-dimensional marginal densities are continuous, then, $E[\hat b_{add}(\x)]$ is (Lipschitz-) continuous on $[-1,1]^d$. That is,
\[
\left| E[\hat b_{add}(\x)] - E[\hat b_{add}(\y)]\right| \leq L||\x - \y||\,.
\]
\end{cor}
\paragraph{ Proof of Corollary~\ref{cor:addprod}} We have
\[
E[s^0(\x)] = \sum_{l=1}^d \int r(\x)f(\x)\,d\x_{-l} - \frac{d-1}{2}\int r(\x)f(\x)\,d\x +o_p(1) \,.
\]
Observe that
\[
r(\x) = x_jx_k = (x_j-\bar X_j)(x_k-\bar X_k) + x_j\bar X_k + x_k\bar X_j + \bar X_j\bar X_k \,.
\]
Because of linearity of the rest of terms, conditionally on $\X$, we only need to consider the first term for bias. Hence, if $r(\x)=(x_j-\Bar X_j)(x_k-\bar X_k)$, then
\begin{eqnarray*}
\lefteqn{\sum_{l=1}^d \int r(\x)(f)(\x)\,d\x_{-l}}\\
&=& (x_j-\bar X_j)\int (u_k-\bar X_k)(f)_{j,k}(x_j,u_k)\,du_k + (x_k-\bar X_k)\int (u_j-\bar X_j)(f)_{j,k}(u_j,x_k)\,du_j \\
& & + \sum_{l\neq (j,k)}\int (u_j-\bar X_j)(u_k-\bar X_k)(f)_{j,k,l}(u_j,u_k,x_l)\,du_jdu_k \,.
\end{eqnarray*}

If two and three marginal design densities are independent, then 
\begin{eqnarray*}
E[s^0(\x)] &=& (x_j-\bar X_j)f_j(x_j)(\mu_k-\bar X_k) + (x_k-\bar X_k)f_k(x_k)(\mu_j-\bar X_j)\\
& &+\sum_{l\neq(j,k)}(\mu_j-\bar X_j)(\mu_k-\bar X_k)f_l(x_l) - \frac{d-1}{2}(\mu_j-\bar X_j)(\mu_k-\bar X_k) + o_p(1) \,,
\end{eqnarray*}
where $\mu_j = \int u_jf_j(u_j)\,du_j$.

Because $\mu_j-\bar X_j = O_p(n^{-1/2})$, 
\[
E[s^0(\x)] = o_p(1) \,,
\]
which leads to
\[
B_{add}(\x) = -(x_j-\bar X_j)(x_k-\bar X_k)+o_p(1) \,.
\] 
In particular when $\x = \bar\X$, 
\[
B_{add}(\bar\X) = o_p(1) \,.
\]

Moreover, for general $\x$, we have
\begin{eqnarray*}
\lefteqn{[B_{add}(\x)-B_{add}(\y)]} \\
&=& (x_j-y_j)\int (u_k-\bar X_k)f_{j,k}(x_j,u_k)\,du_k + (y_j-\bar X_j)\int (u_k-\bar X_k)\{f_{j,k}(x_j,u_k)-f_{j,k}(y_j,u_k)\}\,du_k \\
& & + (x_k-y_k)\int (u_j-\bar X_j)f_{j,k}(u_j,x_k)\,du_j + (y_k-\bar X_k)\int (u_j-\bar X_j)\{f_{j,k}(u_j,x_k)-f_{j,k}(u_j,y_k)\}\,du_j \\
& & + \sum_{l\neq(j,k)}\int (u_j-\bar X_j)(u_k-\bar X_k)\{(f)_{j,k,l}(u_j,u_k,x_l)-(f)_{j,k,l}(u_j,u_k,y_l)\}\,du_jdu_k + o_p(1) \,.
\end{eqnarray*}

Because $f$ is continuous, it follows that
\[
|B_{add}(\x) - B_{add}(\y)| \leq C||\x - \y||\,,
\]
where $C$ is a constant.
Therefore, for a constant $L$, with probability tending to 1, $E\hat r_{add}(\x)$ is continuous, that is,
\[
|E\hat r_{add}(\x) - E\hat r_{add}(\y)| \leq L||\x-\y|| \,.
\]
for a constant $L$.

\subsubsection{ Fixed uniform design }

We first focus on the case where $r''$ exists and continuous. Because the estimator is linear, it is enough to consider the bilinear function.

\begin{lem} \label{lem:naddladd}
Suppose that 
\[
b(\u) = (u_j-\bar U_j)(u_k-\bar U_k)
\]
and let $\hat b_{add,w}(\u)$ be additive estimator based on design density $f_w=\tilde f$ given in ~(\ref{e:ftilde}).
If design points are equally spaced, then
\[
E[\hat b_{add,w}(\u)] = 0\,.
\]
\end{lem}

\begin{proof}
This follows from Proposition~\ref{prop:nadd}. 
\end{proof}

\begin{prop} \label{prop:unifprod}
If design points are equally spaced, then,
\[
B_{ladd}^{(2)}(\x_0) = -\frac{1}{2}\sum_{j\neq k}r_{j,k}''(\x^*)(\bar X_j^*-x_{0j})(\bar X_k^*-x_{0k}) \,,
\]
where $\bar X_j^*$ is the local average of $\{X_j: X_j \in [x_0-w,x_0+w]\}$ and $\x^* \in [\x_0, \x_0+w\u]$, which is determined by Taylor expansion.
\end{prop}

\paragraph{ Proof of Proposition~\ref{prop:unifprod} } Write
\[
u_ju_k = (u_j-\bar U_j)(u_k-\bar U_k) + (u_j-\bar U_j)\bar U_k + (u_k-\bar U_k)\bar U_j + \bar U_j\bar U_k \,.
\]
By Lemma~\ref{lem:naddladd}, the estimator of the first term is zero.
The additional terms are linear and thus do not add additional bias. Thus,
\[
B[\widehat{\tilde r}_{add}(\u)] = (u_j-\Bar U_j)\bar U_k + (u_k-\bar U_k)\bar U_j + \bar U_j\bar U_k \,,
\]
and, in particular at $\u = \0$, we have
\[
B[\widehat{\tilde r}_{add}(\0)] = -\bar U_j\bar U_k = -\frac{(\bar X_j^*-x_{0j})(\bar X_k^*-x_{0k})}{w^2} \,.
\]
Therefore,
\[
B_{ladd}^{(2)}(\x_0) = \frac{w^2}{2}\sum_{j\neq k}^dr_{j,k}''(\x)B[\widehat{\tilde r}_{add}(\0)] =  -\frac{1}{2}\sum_{j\neq k}r_{j,k}''(\x)(\bar X_j^*-x_{0j})(\bar X_k^*-x_{0k}) \,.
\]

\paragraph{ Proof of Proposition~\ref{prop:ladd_unif_b2}}

\begin{eqnarray*}
\hat r_{ladd}(\x_0) &=& \frac{w^4}{4!}\Big\{2\sum_{l=1}^d\sum_{j\neq k}\Big(\frac{1}{n}\sum_{i=1}^n r_{j,j,k,k}''''(\x_0) U_{i,j}^2U_{i,k}^2K_{\tilde h_l}(U_{i,l},0)\Big) \\
& &  - (d-1)\Big(\sum_{j\neq k}\frac{1}{\tilde n}\sum_{i=1}^{\tilde n} r_{j,j,k,k}''''(\x_0)U_{i,j}^2U_{i,k}^2\Big)\Big\}(1+o_p(1))
\end{eqnarray*}

\begin{eqnarray*}
& & E[U_jU_kK_{\tilde h_l}(U_l,u_l) \\
& & E[U_jU_kU_lU_mK_{\tilde h_j}(U_j,u_j)] = 0 \\
& & E[U_j^2U_kU_lK_{\tilde h_j}(U_j,u_j)] = 0 \\
& & E[U_j^2U_kU_lK_{\tilde h_k}(U_k,u_k)] = 0 \\
& & E[U_j^3U_kK_{\tilde h_l}(U_l,u_l)] = 0 \\
& & E[U_j^3U_kK_{\tilde h_j}(U_j,u_j)] = 0 \\
& & E[U_j^3U_kK_{\tilde h_k}(U_k,u_k)] = 0  \,.
\end{eqnarray*} 
and thus bias is dominated by the term $u_j^2u_k^2$. 
\begin{eqnarray*}
\hat r_{ladd}(\x_0) &=& \frac{w^4}{4!}\Big\{2\sum_{l=1}^d\sum_{j\neq k}\Big(\frac{1}{n}\sum_{i=1}^n r_{j,j,k,k}''''(\x_0) U_{i,j}^2U_{i,k}^2K_{\tilde h_l}(U_{i,l},0)\Big) \\
& &  - (d-1)\Big(\sum_{j\neq k}\frac{1}{\tilde n}\sum_{i=1}^{\tilde n} r_{j,j,k,k}''''(\x_0)U_{i,j}^2U_{i,k}^2\Big)\Big\}(1+o_p(1))
\end{eqnarray*}

\begin{eqnarray*}
E[U_j^2U_k^2] &=& \frac{1}{9} \\
E[U_j^2U_k^2K_{\tilde h_l}(U_l,u_l)] &=& \frac{1}{2}E[U_j^2U_k^2] = \frac{1}{18} \\
E[U_j^2U_k^2K_{\tilde h_j}(U_j,u_j)] &=& \frac{1}{3}\Big(\frac{1}{2}u_j^2 + \frac{\tilde h_j^2}{2}\mu_2(K)\Big) \,.
\end{eqnarray*}

\begin{eqnarray*}
\lefteqn{\sum_{l=1}^d\Big(2E\big[\sum_{j\neq k}r_{j,j,k,k}''''(\x_0)U_j^2U_k^2K_{\tilde h_l}(U_l,0)\big]\Big)- (d-1)E[\sum_{j\neq k} r_{j,j,k,k}''''(\x_0)U_j^2U_k^2]} \\
&=& \sum_{j\neq k}r_{j,j,k,k}''''(\x_0)\Big\{E\big[U_j^2U_k^2\big](d-2) + 2E\big[U_j^2U_k^2\big(K_{\tilde h_j}(U_j,0)+K_{\tilde h_k}(U_k,0)\big)\big] \Big\} \\
& & - (d-1)\sum_{j\neq k} r_{j,j,k,k}''''(\x_0)E\big[U_j^2U_k^2\big] \\
&=& \sum_{j\neq k}r_{j,j,k,k}''''(\x_0)\Big\{2E\big[U_j^2U_k^2\big(K_{\tilde h_j}(U_j,0)+K_{\tilde h_k}(U_k,0)\big)\big] - E\big[U_j^2U_k^2\big] \Big\} \\
&=& \sum_{j\neq k}r_{j,j,k,k}''''(\x_0)\Big\{\mu_2(K)\frac{h_j^2+h_k^2}{3w^2} -\frac{1}{9}\Big\}
\end{eqnarray*}

\begin{eqnarray*}
B_{ladd}^{(2)}(\x_0) &=& \frac{w^4}{4!}\Big\{\sum_{j\neq k} r_{j,j,k,k}''''(\x_0)\Big(\frac{1}{9} +  \frac{\tilde h_j^2+\tilde h_k^2}{3} - \frac{2}{9}\Big)\Big\}(1+o_p(1)) \\
&=&  \frac{w^4}{4!}\Big\{\sum_{j\neq k} r_{j,j,k,k}''''(\x_0)\Big(-\frac{1}{9} +  \frac{\tilde h_j^2+\tilde h_k^2}{3}\Big)\Big\}(1+o_p(1)) \\
&=&  \Big\{-\frac{w^4}{4!\cdot 9}\Big\{\sum_{j\neq k} r_{j,j,k,k}''''(\x_0) + \frac{w^2}{4!\cdot 3}\sum_{j\neq k} r_{j,j,k,k}''''(\x_0)(h_j^2+h_k^2) \Big\}(1+o_p(1))
\end{eqnarray*}
As $h/w \rt 0$, only the first term is needed.

\paragraph{ Alternative Proof of Lemma~\ref{lem:laddprod} } Note that $\P_{add}\tilde r = 0$ and $\tilde r_{k,k}''(\u) = 0, k=1,\cdots,d$. 
By Proposition~\ref{cor:badd}, with $\tilde r$, $\tilde f$ and $\tilde h$ replaced, bias can be expressed as 
\[
2^{d-1}B_{add}(\0) = 2^{d-1}B_{add}(\u)\big|_{\u=0} \,,
\]
where
\begin{eqnarray*}
\lefteqn{2^{d-1}B_{add}(\u)} \\
&=&\Big[\sum_{l=1}^d \Big(\int \tilde r(\u)\tilde f(\u)\,d\u_{-l} + \frac{\mu_2(K)}{2}\sum_{m=1}^d\tilde h_m^2 \int\tilde r(\u) \tilde f_{m,m}''(\u)\,d\u_{-l} \\
& & + \mu_2(K)\sum_{m=1}^d\tilde h_m^2\int\tilde r_m'(\u)(\tilde f)_m'(u_m)\,d\u_{-l}\Big) \\
& & -\frac{d-1}{2}\Big(\int \tilde r(\u) \tilde f(\u)\,d\u + \frac{\mu_2(K)}{2}\sum_{m=1}^d \tilde h_m^2\int\tilde r(\u)\tilde f_{m,m}''(\u)\,d\u \\
& & - \mu_2(K)\sum_{m=1}^d\tilde h_m^2\int\tilde r_m'(\u)(\tilde f)_m'(u_m)\,d\u\Big)\Big]\big(1+o_p(1)\big) \,.
\end{eqnarray*}

From~(\ref{e:ftilde}), partial derivatives of $\tilde f$ may be written as
\begin{eqnarray*}
\frac{\partial}{\partial u_j}(\tilde f)(\u) &=& \frac{1}{2^d f(\x_0)} \Big(wf_j'(\x_0) + \frac{1}{2}w^2\sum_{k=1}^d 2f_{j,k}''(\x_0)u_k\Big) + o(w^2) \\
\frac{\partial^2}{\partial u_j^2}(\tilde f)(\u) &=& \frac{w^2}{2^d f(\x_0)}f_{j,j}''(\x_0)  + o(w^2)\\
\frac{\partial^2}{\partial u_j\partial u_k}(\tilde f)(\u) &=& \frac{w^2}{2^{d}f(\x_0)}f_{j,k}''(\x_0) + o(w^2) \,.
\end{eqnarray*}
Note that the second derivatives do not depend on $\u$ and thus
\[
\int\tilde r(\u)\tilde f_{m,m}''(\u)\,d\u_{-l} = 0 \,,\quad\mbox{ for all } l,m = 1,\cdots,d\,.
\]
The other terms are simplified to
\begin{eqnarray*}
\sum_{l=1}^d \int \tilde r(\u)\tilde f(\u)\,d\u_{-l}
&=& u_j\int u_k(\tilde f)_{j,k}(u_j,u_k)\,du_k + u_k\int u_j(\tilde f)_{j,k}(u_j,u_k)\,du_j \\
& & + \sum_{l\neq (j,k)}\int u_ju_k(\tilde f)_{j,k,l}(u_j,u_k,u_l)\,du_jdu_k \,,\\
\int\tilde r(\u)(\tilde f)(\u)\,d\u &=& \int u_ju_k(\tilde f)_{j,k}(u_j,u_k)\,du_jdu_k \,,
\end{eqnarray*}
and
\begin{eqnarray*}
\sum_{l=1}^d\sum_{m=1}^d \int \tilde r_{m}'(\u)\tilde h_{m}^2(\tilde f)_{m}'(u_{m})\,d\u_{-l} 
&=& \sum_{l=1}^d \int u_j \tilde h_k^2(\tilde f)_k(u_k) + u_k\tilde h_j^2(\tilde f)_j(u_j)\,d\u_{-l} \\
&=&  2^{d-2}u_j\int \tilde h_k^2(\tilde f)_k(u_k)\,du_k + 2^{d-2}u_k\int \tilde h_j^2(\tilde f)_j(u_j)\,du_j \,,\\
\sum_{m=1}^d \int \tilde r_m'(\u)\tilde h_m^2(\tilde f)_m'(u_m)\,d\u &=& 0 \,.
\end{eqnarray*}
Thus,
\begin{eqnarray*}
2^{d-1}B_{add}(\0) &=& \sum_{l\neq (j,k)}\int u_ju_k(\tilde f)_{j,k,l}(\u)\,du_jdu_k - \frac{d-1}{2}\int u_ju_k(\tilde f)_{j,k}(\u)\,du_jdu_k \\
&=& \frac{w^2}{18}\Big(\sum_{l\neq (j,k)} \frac{\frac{\partial^2}{\partial u_j\partial u_k}(f)_{j,k,l}(\x_0)}{(f)_{j,k,l}(\x_0)} - (d-1)\frac{\frac{\partial^2}{\partial u_j\partial u_k}(f)_{j,k}(\x_0)}{(f)_{j,k}(\x_0)}\Big)\big(1+o_p(1)\big)
\end{eqnarray*}

By Lemma~\ref{lem:sadd}, $E[\widehat{\tilde r}(\0)] = 2^dB_{add}(\0)$.

}

\noprint{
\newpage
\begin{lem} Let $t(\x) = 2^{d-1}(\S_{add}r)(\x)$. Then, 
\begin{eqnarray*}
\lefteqn{t^0(\x)}\\
&=&\sum_{k=1}^d\Big(\frac{1}{n}\sum_{i=1}^{n} K_{h_k}(X_{i,k},x_k)\Big) r_{add,k}(x_k) \\
& &+\sum_{j\neq k}\int \Big(\frac{1}{n}\sum_{i=1}^{n} K_{h_j}(X_{i,j},u)K_{h_k}(X_{i,k},x_k)\Big)r_{add,j}(u)\,du \\
& &-\frac{d-1}{2}\sum_{j=1}^d\int \Big(\frac{1}{n}\sum_{i=1}^{n} K_{h_j}(X_{i,j},u)\Big)r_{add,j}(u)\,du \\
& &+\sum_{k=1}^d\Big(\frac{1}{n}\sum_{i=1}^{n} K_{h_k}(X_{i,k},x_k)\Big(\frac{X_{i,k}-x_k}{ h_k}\Big)\Big)r_{add}^k(x_k) \\
& &+\sum_{j\neq k}\int\Big(\frac{1}{n}\sum_{i=1}^{n}K_{h_j}(X_{i,j},u)K_{h_k}(X_{i,k},x_k)\Big(\frac{X_{i,j}-x_j}{h_j}\Big)\Big)r_{add}^j(u)\,du \\
& &-\frac{d-1}{2}\sum_{j=1}^d\int\Big(\frac{1}{n}\sum_{i=1}^{n}K_{h_j}(X_{i,j},u)\Big(\frac{X_{i,j}-u}{h_j}\Big)\Big)r_{add}^j(u)\,du
\end{eqnarray*}
and for $k=1, \cdots, d$, 
\begin{eqnarray*}
\lefteqn{t^k(\x)} \\  
&=&\sum_{k=1}^d\Big(\frac{1}{n}\sum_{i=1}^{n} K_{h_k}(X_{i,k},x_k)\frac{X_{i,k}-x_k}{h_k}\Big) r_{add,k}(x_k) \\
& &+\sum_{j\neq k} \int \Big(\frac{1}{n}\sum_{i=1}^{n} K_{h_j}(X_{i,j},u)K_{h_k}(X_{i,k},x_k)\frac{X_{i,k}-x_k}{h_k}\Big)r_{add,j}(u)\,du \\
& &+\sum_{k=1}^d\Big(\frac{1}{n}\sum_{i=1}^{n} K_{h_k}(X_{i,k},x_k)\Big(\frac{X_{i,k}-x_k}{h_k}\Big)^2\Big)r_{add}^k(x_k) \\
& &+\sum_{j\neq k}\int\Big(\frac{1}{n}\sum_{i=1}^{n}K_{h_j}(X_{i,j},u)K_{h_k}(X_{i,k},x_k)\frac{X_{i,j}-u}{h_j}\frac{X_{i,k}-x_k}{h_k}\Big)r_{add}^j(u)\,du \,,
\end{eqnarray*}
where $r_{add} = \P_{add}r$ defined as in Lemma~\ref{lem:padd} and $\r_{add}^0 =\sum_{k=1}^d r_{add,k}$. 
\end{lem}

If $f$ is only assumed to be once continuously differentiable, then the order of bias is reduced to
\[
B^2\big(\hat r_{ladd}(\x_0)\big) = \max\{O(h^4), O(w^6)\} \,.
\]
Therefore, we have
\[
MSE = O(h^4 + w^6 + (nw^{d-1}h)^{-1})\,,
\]
Then the optimal relationship between $h$ and $w$ is given by $h \sim w^{3/2}$. As a result, we have
\[
MSE = O\big(w^6 + (nw^{d+1/2})^{-1}\big)\,,
\]
which leads to an optimal choice of $w$ and $h$ as
\[
w \sim n^{-1/(6.5+d)} \,,\quad h \sim n^{-1.5/(6.5+d)} \,,
\]
and 
\[
MSE \sim n^{-6/(6.5+d)} \,.
\]

\subsection{ Proof of Proposition~\ref{prop:laddprod}? }

\begin{itemize}
\item[(iii)] $f$ is twice continuously differentiable \\
\begin{eqnarray*}
\lefteqn{(a) = E[U_jU_kK_{\tilde h_l}(U_l,x_l)]} \\
&=& \int u_ju_kK_{\tilde h_l}(u_l,x_l)\tilde f_{j,k,l}(u_j,u_k,u_l)\,du_jdu_kdu_l \\
&=& \int u_ju_k\Big\{\!\int\!K(u_l)\frac{1}{2^3f(\x_0)}\Big(f(\x_0)+w\big(f_j'(\x_0)u_j + f_k'(\x_0)u_k \\
& & + f_l'(\x_0)(x_l + \tilde h_lu_l)\big) + O(w^2)\Big)\,du_l\Big\}\,du_jdu_k \\
&=& O(w^2) \\
\lefteqn{(b) = E[U_jU_k K_{\tilde h_k}(U_k,x_k)] } \\
&=& \int u_j\Big\{\!\int\!(x_k+\tilde h_lu_k)K(u_k)\frac{1}{4f(\x_0)}\Big(f(\x_0)+w\big(f_j'(\x_0)u_j \\
& &  + f_k'(\x_0)(x_k+\tilde h_ku_k)\big) + O(w^2)\Big)\,du_k\Big\}\,du_j \\
&=& \frac{wf_j'(\x_0)}{4f(\x_0)}\int u_j^2\,du_j + O(w^2) \\
&=& \frac{wx_k}{6}\frac{\frac{\partial}{\partial u_j}f_{j,k}''(\x_0)}{f_{j,k}(\x_0)} + O(w^2) \\
\lefteqn{(c) = E[U_jU_k] = O(w^2)} 
\end{eqnarray*}
\item[(ii)] $f$ is once continuously differentiable \\
\[
\tilde f(\u) = \frac{1}{2^d} + O(w) \,,
\]
and then
\[
(a) = O(w)\,,\quad (b) = O(w) \,,\quad (c) = O(w) \,.
\]
\item[(i)] $f$ is uniform \\
\[
\tilde f(\u) = \frac{1}{2^d} \,,
\]
and then
\[
(a)=(b)=(c)=0 \,.
\]
\end{itemize}
Therefore, the result follows.

\subsection{General design}

Assuming that $r$ is twice continuously differentiable, the regression function may be expressed as
\[
\tilde r(\u) = r(\x_0+w\u) = \mbox{additive part} + \frac{w^2}{2}\sum_{j\neq k} r_{j,k}''(\x^*)u_ju_k\,.
\]
If $r''$ is Lipschitz continuous, then
\[
|r_{j,k}''(\x^*) - r_{j,k}''(\x_0)| \leq C|\x^*-\x_0| = \tilde C W||\1|| \,,
\]
where $C$ and $\tilde C$ are constants. This implies that
\[
r_{j,k}''(\x^*) = r_{j,k}''(\x_0) + O(w)
\]
and thus we
Thus, for general density case, nonadditive bias is dominated by the term $u_ju_k$,
compared to the term $u_j^2u_k^2$ for uniform density. Because $r_{j,k}''(\x^*)$ is
bounded, it is enough to consider the bilinear function $u_ju_k$. 

\begin{lem} \label{lem:laddprod} Suppose that 
\[
b(\u) = u_ju_k
\]
and let $\hat b_{add,w}(\u)$ be the additive estimator based on design density $f_w=\tilde f(\u)$ given in ~(\ref{e:ftilde}). If $f$ is twice continuously differentiable, then
\begin{eqnarray*}
E[\hat b_{add,w}(\0)] &=& \frac{w^2}{9}\Big(\sum_{l\neq (j,k)} \frac{\frac{\partial^2}{\partial u_j\partial u_k}(f)_{j,k,l}(\x_0)}{(f)_{j,k,l}(\x_0)} - (d-1)\frac{\frac{\partial^2}{\partial u_j\partial u_k}(f)_{j,k}(\x_0)}{(f)_{j,k}(\x_0)}\Big)\big(1+o_p(1)\big)\,.
\end{eqnarray*}
Therefore, the nonadditive bias is given by
\[
B_{ladd}^{(2)}(\x_0) = \frac{w^4}{18}\sum_{j\neq k} r_{j,k}''(\x_0)\frac{\frac{\partial^2}{\partial u_j\partial u_k}(f)_{j,k,l}(\x_0)}{(f)_{j,k,l}(\x_0)} - (d-1)\frac{\frac{\partial^2}{\partial u_j\partial u_k}(f)_{j,k}(\x_0)}{(f)_{j,k}(\x_0)}\Big)\big(1+o_p(1)\big) \,.
\]
\end{lem}
\begin{proof} 
From the equation~(\ref{e:padd_e}), it can be written as
\[
E[\hat b_{add,w}(\u)] = \Big\{2\sum_l\Big(\frac{1}{\tilde n}\sum_i U_{i,j}U_{i,k}K_{\tilde h_l}(U_{i,l},u_l)\Big) - (d-1)\Big(\frac{1}{\tilde n}\sum_i U_{i,j}U_{i,k}\Big)\Big\}(1+o_p(1))\,.
\]
If $f''$ is Lipschitz continuous, then
\begin{eqnarray*}
\lefteqn{\tilde f_{j,k,l}(u_j,u_k,x_l+\tilde h_lu_l)} \\
&=& \frac{1}{2^3f(\x_0)}\Big\{f(\x_0) + w(f_j'(\x_0)u_j + f_k'(\x_0)u_k) \\
& & +\frac{w^2}{2}\big(f_{j,j}''(\x_0)u_j^2 + f_{k,k}''(\x_0)u_k^2 + 2f_{j,k}''(\x_0)u_ju_k \\
& & -\frac{1}{3}(f_{j,j}''(\x_0)+f_{k,k}''(\x_0)+f_{l,l}''(\x_0))\big)+ wf_l'(\x_0)(x_l+\tilde h_lu_l) \\
& & +\frac{w^2}{2}\big(f_{j,l}''(\x_0)u_j(x_l+\tilde h_lu_l)+ f_{k,l}''(\x_0)u_k(x_l+\tilde h_lu_l) + f_{l,l}''(\x_0)(x_l+\tilde h_lu_l)^2\big)\Big\} + O(w^3)\,.
\end{eqnarray*}
\begin{eqnarray*}
\lefteqn{E[U_jU_kK_{\tilde h_l}(U_l,x_l)]} \\
&=& \frac{w^2f_{j,k}''(\x_0)}{8f(\x_0)}\int u_j^2u_k^2\,du_jdu_k + O(w^3) \\
&=& \frac{w^2}{18}\frac{\frac{\partial^2}{\partial u_j\partial u_k}(f)_{j,k,l}(\x_0)}{(f)_{j,k,l}(\x_0)} + O(w^3) \\
\lefteqn{E[U_jU_kK_{\tilde h_k}(U_k,x_k)]} \\
&=&\int u_j\Big\{\!\int\! (x_k+\tilde h_ku_k)\frac{1}{4f(\x_0)}\Big(f(\x_0)+ w\big(f_j'(\x_0)u_j + f_k'(\x_0)(x_k+\tilde h_ku_k)\big) \\
& & + \frac{w^2}{2}\big(f_{j,j}''(\x_0)u_j^2 + 2f_{j,k}''(\x_0)u_j(x_k+\tilde h_ku_k) + f_{k,k}''(\x_0)u_k^2 \\
& & - \frac{1}{3}(f_{j,j}''(\x_0)+f_{k,k}''(\x_0))\big)\Big) + O(w^3) \\
&=& \frac{1}{f(\x_0)}\Big(\int u_j^2 \big(wx_kf_j'(\x_0) + w^2x_k^2f_{j,k}''(\x_0)\big)\,du_j\Big) + O(w^3)\\
&=& \frac{1}{6(f)_{j,k}(\x_0)}\Big(wx_k\frac{\partial}{\partial u_j}(f)_{j,k}(\x_0)+w^2x_k^2\frac{\partial^2}{\partial u_j\partial u_k}(f)_{j,k}(\x_0)\Big) + O(w^3) \\
\lefteqn{E[U_jU_k]} \\
&=& \frac{w^2}{9(f)_{j,k}(\x_0)}\frac{\partial^2}{\partial u_j\partial u_k}(f)_{j,k}(\x_0) + O(w^3)
\end{eqnarray*}
Because we estimate at $\u=0$, we have
\[
E[\hat b_{add,w}(\0)] = \frac{w^2}{9}\Big(\sum_{l\neq (j,k)} \frac{\frac{\partial^2}{\partial u_j\partial u_k}(f)_{j,k,l}(\x_0)}{(f)_{j,k,l}(\x_0)} 
 - (d-1)\frac{\frac{\partial^2}{\partial u_j\partial u_k}(f)_{j,k}(\x_0)}{(f)_{j,k}(\x_0)}\Big)(1+o_p(1)) \,.
\]
\end{proof}

Combined with Lemma~\ref{lem:ladd_badd}, we have
\[
B^2[\hat r_{ladd}(\x_0)] = \max\{O(h^4), O(w^8)\} \,,
\]
and Corollary~\ref{cor:cor1} and Proposition~\ref{prop:ladd_b2_var} are proved.

Similarly, the coefficients for the optimal smoothing parameters are derived.

\paragraph{ Proof of Proposition~\ref{prop:ladd_w}}[???]
\noprint{
\begin{prop} 
Assume that $f$ is twice continuously differentiable.
Assume that $h_1=\cdots=h_d$ and let $h = C_hw^2$. The smoothing parameter $w$ that minimizes asymptotic MSE is given by
\[
w \sim \Big(\frac{C_1}{C_2}\Big)^{\frac{1}{9+d}} n^{-\frac{1}{9+d}} \,,
\]
where
\begin{eqnarray*}
C_1 &=& \mu_0(K^2)\sigma^2 d(d+1) \\
C_2 &=& 4C_h\Big\{\frac{C_h^2\mu_2(K)}{2}\sum_{j=1}^d r_{j,j}''(\x_0) + \frac{1}{18}\sum_{j\neq k} r_{j,k}''(\x_0)\Big(\sum_{l\neq (j,k)} \frac{\frac{\partial^2}{\partial u_j\partial u_k}(f)_{j,k,l}(\x_0)}{(f)_{j,k,l}(\x_0)} \\
& & - (d-1)\frac{\frac{\partial^2}{\partial u_j\partial u_k}(f)_{j,k}(\x_0)}{(f)_{j,k}(\x_0)}\Big)\Big\}^2 \,.
\end{eqnarray*}
\end{prop}
}

\section*{Appendix: Proof of Proposition~\ref{prop:aic}}

\begin{lem}
\[
tr((H^\prime H)^\prime(H^\prime H))=O\big(tr(H^\prime H)\big)=O\big(1/(w^{d-1}h)\big)
\]
\end{lem}

{\bf Proof:} Denote by $H_i$ the hat matrix of the additive estimator used for local
additive estimation at $x_0=X_i$.
Then, inflating the matrix to an $n\times n$ matrix, the $i$th line of $H$ is the $i$th
line of $H_i$ := $H_{i,i}$\,.
Now, considering the form of the estimator $\hat{r}_i = \hat{r}_0+ \hat{r}_1+ \ldots +
\hat{r}_d$, where all components are oracle, we have 
\[
H_{i,j} = \left\{\begin{array}{cc}
O(\frac{1}{\tilde n\tilde h}) & \text{ if for all } k:|X_{i_k}-X_{j_k}|\leq w \text{ and
for some } k: |X_{i_k}-X_{j_k}|\leq h \\
O(\frac{1}{\tilde n}) & \text{ if for all } k: |X_{i_k}-X_{j_k}|\leq w \text{ and for 
all } k: |X_{i_k}-X_{j_k}|> h\\
0 & \text{ otherwise}\end{array}\right.
\]
Note, that these $O()$s are uniform over $X$ because of (A.5), using Gao (2003) as in
Studer et al. (2005).
Let's first look at $tr(H'H)$.
We have 
\[
H_{i,i}' H_{i,i} = O\left(O(\frac{1}{\tilde n\tilde h})^2  O(\tilde n\tilde h) +
O(\frac{ 1}{ \tilde n})^2 O(\tilde n)\right) = O\left(\frac{1}{\tilde n\tilde h} \right) =
O\left(\frac{1}{nw^{d-1}h} \right). 
\]
Consequently, 
\[ 
tr(H'H) =  O\left(\frac{1}{w^{d-1}h} \right). 
\]

Now, look at the general elements of $H'H$.
With a slight abuse of notation, 
\begin{eqnarray*} 
H_{i,i}'H_{j,j} &=& O\left(O(\frac{1}{\tilde n\tilde h})^2 
O(\tilde n\tilde h) + O(\frac{1}{\tilde n\tilde h})O(\frac{ 1}{ \tilde n})  O(\tilde
n\tilde h) + O(\frac{ 1}{ \tilde n})^2 O(\tilde n)\right) \\
& &\qquad\text{ if for all } k: (X_{i_k}\pm w)\cap (X_{j_k}\pm w) \not= \varnothing \\
 & &\qquad\text{ and for some } k:   (X_{i_k}\pm h)\cap (X_{j_k}\pm h)\not= \varnothing 
\\
& & O\left(O(\frac{1}{\tilde n\tilde h})O(\frac{ 1}{ \tilde n})  O(\tilde n\tilde h) +
O(\frac{ 1}{ \tilde n})^2 O(\tilde n)\right)\\
& & \qquad\text{ if for all } k:(X_{i_k}\pm w)\cap (X_{j_k}\pm w) \not= \varnothing \\
 & & \qquad\text{ and for some } k:   (X_{i_k}\pm w)\cap (X_{j_k}\pm h)\not= \varnothing 
\\
 & & \qquad\text{ or }   (X_{i_k}\pm h)\cap (X_{j_k}\pm w)\not= \varnothing  \\
 & & O\left(O(\frac{ 1}{ \tilde n})\right) \qquad\text{ if for all } k: (X_{i_k}\pm
w)\cap (X_{j_k}\pm w) \not= \varnothing \\
& & \qquad \text{ and } (X_{i_k}\pm w)\cap (X_{j_k}\pm h)= \varnothing \text{ and }
(X_{i_k}\pm h)\cap (X_{j_k}\pm w)= \varnothing  \\
 & & 0 \quad\text{ otherwise} \,.
\end{eqnarray*} 
Finally,
\begin{eqnarray*} 
H_{i,i}'H_{j,j} &=& O(\frac{1}{\tilde n\tilde h})   \text{ if for all } k: (X_{i_k}\pm
w)\cap (X_{j_k}\pm w) \not= \varnothing \\
 & & \qquad \text{ and for some } k:   (X_{i_k}\pm h)\cap (X_{j_k}\pm h)\not= \varnothing 
\\
 & & O(\frac{ 1}{ \tilde n})  \text{ if for all } k: (X_{i_k}\pm w)\cap (X_{j_k}\pm w)
\not= \varnothing  \\
& & \qquad \text{ and for all } k:   (X_{i_k}\pm h)\cap (X_{j_k}\pm h)= \varnothing \\
 & & 0 \text{ otherwise}\,.
\end{eqnarray*} 
Thus, $H'H$ has the same structure as $H$, of course with different constants and larger
non-zero regions, but all of the same order. Therefore, 
\[
tr(H'HH'H) = O(tr(H'H)) = O\left(\frac{1}{w^{d-1}h} \right). 
\]

\begin{lem}
\[
\frac{1}{n\sigma^2}\verr^\prime(I-H)^\prime(I-H)\r = \frac{1}{n\sigma^2}<(I-H)\verr,
(I-H)\r> = O_p\big(\frac{h^2+w^4}{\sqrt{n}}\big)
\]
\end{lem}

{\bf Proof:} First note that $(I-H)\r  = O((h^2+w^4)\one)$. 
\begin{eqnarray*}
\lefteqn{\frac{1}{n\sigma^2}<(I-H)\verr, (I-H)\r>} \\ 
&\leq& \frac{1}{n\sigma^2}||(I-H)\verr||||(I-H)\r|| \\
&\leq& \frac{1}{n\sigma^2}O_p(\sqrt{n})O((h^2+w^4)\one) =
O_p\big(\frac{h^2+w^4}{\sqrt{n}}\big)
\end{eqnarray*}
where the last inequality follows from
\[
||(I-H)\verr||^2 = O_p(tr(V[||(I-H)\verr||^2)) = O_p(tr((I-H)(I-H)^\prime\sigma^2) =
O_p(n) \,.
\]

\begin{lem}
\[
V[\verr^\prime H\verr] = V[<\verr, H\verr>] = O(E[||H\verr||^2]) 
\]
If $tr((H^\prime H)(H^\prime H))=O(tr(H^\prime H)$, then 
\[
V[\verr^\prime H^\prime H\verr] = V[||H\verr||^2]= O(E[||H\verr||^2])
\]
\end{lem}

{\bf Proof:} Similar to the proof of lemma 5 in the appendix of Studer et al. (2005), we
use the following fact: For symmetric matrices $B$ and $C$ and $E[\verr^4] =
(3+\kappa)\sigma^4$,
\[
Cov(\verr^\prime B\verr, \verr^\prime C\verr) = 2\sigma^4tr(BC) + \kappa\sigma^4tr(B\cdot 
diag(C))
\]
Putting $B=C=\frac{1}{2}(H+H^\prime)$ gives
\begin{eqnarray*}
V[\frac{1}{2}\verr^\prime(H+H^\prime)\verr] &=& 
2\sigma^4(\frac{1}{4}tr(HH+2H^\prime H+H^\prime H^\prime) + \kappa\sigma^4
tr(\frac{1}{4}(H^\prime+H) diag(H+H^\prime)) \\
V[\verr^\prime H\verr] &=& \sigma^4(tr(HH + H^\prime H) + \kappa\,tr(diag(H)^2) \\
&\leq& \sigma^4(tr(HH)+tr(H^\prime H) + |\kappa| tr(H^\prime H)
\end{eqnarray*}
Using the equivalence of the trace to Hilbert-Schmidt norm,
\[
||H||_{HS}^2 = tr(H^\prime H)
\]
it follows that
\[
tr(HH) = <H^\prime, H>_{HS} \leq ||H^\prime||_{HS}||H||_{HS} = ||H||_{HS}^2 =
tr(H^\prime H) \,.
\]
Hence,
\begin{eqnarray*}
V[\verr^\prime H\verr] &\leq& \sigma^4(2tr(H^\prime H)+|\kappa| tr(H^\prime H) ) \\
&=& \sigma^4(2+|\kappa|) tr(H^\prime H) \\
&=& \sigma^2(2+|\kappa|) E[||H\verr||^2]
\end{eqnarray*}
Thus, $V[<\verr, H\verr>] = O(E[||H\verr||^2])$. 
Moreover, replacing $H$ by $H^\prime H$ in the above leads to
\begin{eqnarray*}
V[\verr^\prime H^\prime H\verr] &\leq& \sigma^2(2+|\kappa|)
tr((H^\prime H)^\prime(H^\prime H)) \,.
\end{eqnarray*} 

Hence, if $tr((H^\prime H)^\prime(H^\prime H))=O(tr(H^\prime H))$, then $V[||H\verr||^2]
= O(tr(H^\prime H)) = O(E[||H\verr||^2])$.

\begin{lem}\label{lem:supnorm}
\[
||H^\prime H||_{\sup}=O_p(1)
\]
\end{lem}

{\bf Proof:} With probability tending to 1, $\S_{add}^{-1}$ is continuous, thus, $\hat
r_{add}$ is stable. Therefore $||H^\prime H||_{\sup} = O_p(1)$.

{\bf Comment:} $H = \S_{add}^{-1}\P_{add}\S_*$ so the additive estimator $H\Y$ is a
sub-vector of $\S_{add}^{-1}$ times some projection. Thus, bounded eigenvalues of
$\S_{add}^{-1}$ ensure that $H$ has bounded eigenvalues. Now we want to show that
$H_{ladd}$ has bounded maximal eigenvalue. Note that the $i$th row of $H_{ladd}$ is the
$i$th row of an additive estimator, $H^{(i)}$, applied to a sub-vector of $\Y$. As a
consequence, the maximal eigenvalue of $H_{ladd}$ is bounded by the maximum of all those
$H^{(i)}$s. But all those additive estimators converge to a limiting operator.

{\bf Comment:} I don't quite understand how the argument of inverse operator is suddenly
connected to the supreme norm of $H^\prime H$, which I believe is the leading
eigenvalue of the matrix $H^\prime H$. I can see that $H$ can be viewed as a function of
the inverse operator but has more component than that. So where does $H^\prime H$ come
from? Why is this true? Would the same be true for local additive estimator?

\begin{lem}
\[
\frac{1}{n\sigma^2}\verr^\prime(I-H)^\prime(I-H)\r = \frac{1}{n\sigma^2}<(I-H)\verr,
(I-H)\r_{add}> = O_p\big(\frac{h^2}{\sqrt{n}}\big)
\]
\end{lem}

{\bf Proof:} First note that $(I-H)\r_{add}  = O(h^2\one)$. 
\begin{eqnarray*}
\lefteqn{\frac{1}{n\sigma^2}<(I-H)\verr, (I-H)\r_{add}>} \\ 
&\leq& \frac{1}{n\sigma^2}||(I-H)\verr||||(I-H)\r_{add}|| \\
&\leq& \frac{1}{n\sigma^2}O_p(\sqrt{n})O(h^2\one) = O_p\big(\frac{h^2}{\sqrt{n}}\big)
\end{eqnarray*}
where the last inequality follows from
\[
||(I-H)\verr||^2 = O_p(tr(V[||(I-H)\verr||^2)) = O_p(tr((I-H)(I-H)^\prime\sigma^2) =
O_p(n) \,.
\]

{\bf Comment:} For local additive estimator with general regression function, the term
will be replaced by
\[
\frac{1}{n\sigma^2}<(I-H)\verr, (I-H)\r> = O_p\big(\frac{h^2+w^4}{\sqrt{n}}\big) \,.
\]

}

\end{document}